\documentclass[11pt]{article}
\usepackage[utf8]{inputenc}
\usepackage[english]{babel}
\usepackage{amsmath,amsfonts,amsthm,graphicx,float}
\usepackage{natbib}

\usepackage{url}

\usepackage{listings}
\usepackage{appendix}

\usepackage[flushleft]{threeparttable}

\DeclareMathAlphabet{\mathpzc}{OT1}{pzc}{m}{it}

\newcommand{\E}  {\mbox{E}}

\newtheorem{lema}{Lemma}

\usepackage{algorithm}

\usepackage[noend]{algpseudocode}

\def\b1{{\mathbf 1}}

\usepackage[font=small]{caption}

\newtheorem{theorem}{Theorem}

\usepackage{titlesec}

\titleformat*{\section}{\normalfont\fontsize{14}{17}\bfseries}
\titleformat*{\subsection}{\normalfont\fontsize{12}{15}\bfseries}

\usepackage{bm}
\usepackage{enumerate}

\usepackage{soul}
\usepackage[usenames, dvipsnames]{color}

\def\Var{{\rm var}}

\def\nti{$n\rightarrow\infty$}

\def\amse{\textsc{AMSE}}

\usepackage{soul}

\def\beqn{\begin{eqnarray*}}
\def\eeqn{\end{eqnarray*}}
\def\beq{\begin{eqnarray}}
\def\eeq{\end{eqnarray}}

\def\ld{\ldots}
\def\bi{\begin{itemize}}
\def\ei{\end{itemize}}
\def\ra{\rightarrow}

\newcommand{\twofig}[2]{
\hbox to\hsize{\hss
    \vbox{\psfig{figure=#1,height=2.8in,width=2.4in}}\qquad
    \vbox{\psfig{figure=#2,height=2.8in,width=2.4in}}
    \hss}
}
\def\var{{\rm var}}
\def\cov{{\rm cov}}
\def\CV{ross-validation}

{
\theoremstyle{plain}
\newtheorem{assumption}{Assumption}
}

\DeclareMathOperator*{\argmin}{arg\,min}

\date{}

\begin{document}

\title{\LARGE Bagging cross-validated bandwidths with application to Big Data
\thanks{%
{This is a pre-copyedited, author-produced version of an article accepted for publication in Biometrika
following peer review. The version of record of: D Barreiro-Ures, R Cao, M Francisco-Fern\'andez, J D Hart, Bagging
cross-validated bandwidths with application to big data, Biometrika, Volume 108, Issue 4, December 2021, Pages 981--988, https://doi.org/10.1093/biomet/asaa092, published by Oxford University Press, is available online at: https://
doi.org/10.1093/biomet/asaa092. }}}
\author{Daniel Barreiro-Ures \\
Universidade da Coru\~{n}a\thanks{%
Research group MODES, CITIC, Department of Mathematics, Faculty of Computer Science, Universidade da Coru\~na, Campus de Elvi\~na s/n, 15071,
A Coru\~na, Spain}
\and
Ricardo Cao \\
Universidade da Coru\~{n}a\footnotemark[1]
\and %
Mario Francisco-Fern\'andez\\
Universidade da Coru\~{n}a\footnotemark[1]
\and
Jeffrey H. Hart \\
Texas A\&M University\thanks{Department of Statistics, Texas A\&M University, College Station TX 77843, U.S.A.}
}
\maketitle


\begin{abstract}
\cite{HR} proposed and analyzed the use of bagged
  cross-validation to choose the bandwidth of a kernel density
  estimator. They established that bagging greatly reduces the noise
  inherent in ordinary cross-validation, and hence leads to a more
  efficient bandwidth selector. The asymptotic
  theory of \cite{HR} assumes that $N$, the number of bagged
  subsamples, is $\infty$. We expand upon their
  theoretical results by allowing $N$ to be finite, as it is in
  practice. Our results indicate an important difference in the rate
  of convergence of the bagged cross-validation bandwidth for the cases
  $N=\infty$ and $N<\infty$. Simulations quantify the
  improvement in statistical 
  efficiency and computational speed that can result from using bagged
  cross-validation as opposed to a binned implementation of ordinary
  c\CV.  The performance of the
  bagged bandwidth is also illustrated on a real, very large, data
  set. Finally, a byproduct of our study
  is the correction of errors appearing in the 
  \cite{HR} expression for the asymptotic mean squared error of the
  bagging selector.  

\end{abstract}	


\section{Introduction} \label{intro}
C\CV~is a rough-and-ready method of model selection that predates an
early exposition of the method by \cite{Stone}.  In its simplest form,
cross-validation consists of dividing one's data set into two parts,
using one part to build one or more models, and then 
predicting the data in the second part with the models so-built. In this
way, one can objectively compare the predictive ability of different
models. The leave-one-out version of c\CV~is somewhat more involved. 
It excludes one datum from the data set, fits a model from the
remaining observations, uses this model to predict the datum left out, 
and then repeats this process for all the data. 
 
While leave-one-out c\CV~is a very useful method, due in no small part
to its wide applicability, it does have its drawbacks. In the context
of smoothing parameter selection for function estimation, it has
 been regarded skeptically for many years owing to its large
variability \citep[see, e.g.][]{ParMar}. A number of modified versions
of c\CV~have been proposed in an effort to produce more stable
smoothing parameter selectors. These include partitioned
c\CV~\citep{Marron1,BhatHart}, proposals of \cite{Stute} and
\cite{FelKor}, smoothed c\CV~\citep{hmp}, one-sided
c\CV~\citep{HartYi, mmns2009}, a bagged version of
  c\CV~\citep{HR}, indirect c\CV~\citep{SHS} and DO-validation
\citep{MMNS}.

The current paper revisits the application of bagging to the selection
of a kernel density estimator's bandwidth. Given a random sample 
of size $n$ from an unknown density $f$, bagging consists of selecting
$N$ subsamples of size $m<n$, each without replacement, from the $n$
observations.  One then computes a c\CV~bandwidth from each of the $N$
subsets, averages them, and then scales the average down appropriately
to account for the fact that $m<n$. It is well-known that the
use of bagging can lead to substantial reductions in the variability
of an estimator that is nonlinear in the observations
\citep[see][]{FriedHall}. Indeed, this is true in the bandwidth
selection problem, as demonstrated in \cite{HR}.

A method closely related to bagging is
partitioned c\CV~\citep{Marron1}, wherein the data set is partitioned
into mutually exclusive subsets, and a bandwidth is computed from
each subset. One may then average these bandwidths and rescale as in
bagging. A little thought reveals that the statistical properties of
bagging and a 
replicated version of partitioned c\CV~are essentially equivalent, and
hence to fix ideas we consider only bagging in this paper. 

Two other popular methods of bandwidth selection are the plug-in
  method of \cite{SheJo} and the bootstrap \citep{cao1993}.  
  It is worth mentioning that bagged versions
of these two methodologies could also be considered. Some readers
might argue that plug-in methods are more efficient 
  than any version of cross-validation and hence should be the method
  of choice. However, \cite{Loader} challenges this notion and provides good
  reasons for not discarding cross-validatory methods.

The main contributions of our paper are as follows: 

(i) In the case $N=\infty$, we
  provide a correct expression for the 
  asymptotic mean squared error of the bagged bandwidth. The analogous
  expression given by \cite{HR} is in error. Their variance
  approximation is of too large an order, thus downplaying the actual
  reduction in 
variance that is possible with the use of bagging. In addition, we
provide an expression for the first order bias of the bagged
  bandwidth and show that the \cite{HR} 
  bias approximation is actually of smaller order in
  terms of sample size.
  The same bias error appears in the article of \cite{Marron1}. See Appendix 4.

(ii)  We provide a first order
approximation to the variance in the case where $N$ is
finite, which, of course, is the case in practice. This is important
because even if $N=n$,  
the asymptotic variance of the bagged bandwidth is of a
different order than it is when $N=\infty$. The relevance of this
result is immediate for massive data sets, since in such cases 
taking $N$ as large as $n$ can be prohibitive computationally.

(iii) We provide an automatic method to estimate the best subsample size, in the sense
of minimum mean squared error. 

(iv) Both the automatic
method and the bagged bandwidth selector have been implemented into an
R \citep{Rsoft} package, called \texttt{baggedcv}
\citep{baggedcv}, which is already available at CRAN.

\section{Methodology}
\label{sec:meth}

Let $X_1,\ld,X_n$ be a random sample from a density $f$, and consider 
estimating $f(x)$ by the kernel estimator \citep{Parzen1962,Rosenblatt1956}
$$
\hat f_h(x)=\frac{1}{nh}\sum_{i=1}^n K\left(\frac{x-X_i}{h}\right),
$$ 
where $K$ is a symmetric kernel function and $h > 0$ is the bandwidth or smoothing parameter. Making a good choice of the bandwidth is crucial to obtaining a good density estimate. An
oft-used criterion for defining a good bandwidth is based on mean
integrated squared error (MISE), defined by:
$$
M(h)=E\left[\int_{-\infty}^\infty\left\{\hat f_h(x)-f(x)\right\}^2\,{\rm d}x\right].
$$
Suppose that $f$ has two continuous derivatives. As shown by, for
example, \cite{Silverman86}, the minimizer, $h_{n0}$, of $M(h)$
with respect to $h$ is asymptotic to $h_{na}=Cn^{-1/5}$, as
$n\rightarrow\infty$, where 
\beq\label{eq:Cconst}
C=\left\{\frac{R(K)}{\mu_2(K)^2R(f'')}\right\}^{1/5},
\eeq
$R(g)=\int g^2(x)\,{\rm d}x$ and $\mu_j(g) = \int x^jg(x)\, {\rm d}x$ ($j=0,1, \ldots$), 
provided that these integrals exist finite. Ideally, one would use
$h_{n0}$ as a bandwidth in the estimator 
$\hat f_h$, but of course $h_{n0}$ depends on $f$ and so this is not
feasible. A means of estimating $h_{n0}$ is based on c\CV.

The leave-one-out c\CV~criterion can be written as:
$$
CV(h)=\int_{-\infty}^\infty\hat f_h(x)^2\,{\rm d}x-\frac{2}{n}\sum_{i=1}^n \hat
f_h^i(X_i),\quad h>0,
$$
where $\hat f_h^i$ is a kernel estimate computed with the $n-1$
observations other than $X_i$. It is easily shown that, for any $h>0$,
$CV(h)$ is an unbiased estimator of $M(h)-R(f)$. It seems 
natural then to estimate $h_{n0}$ by $\hat h_n$, the minimizer of
$CV(h)$. \cite{HallMarron} show that 
\beq\label{cvrate}
n^{1/10}\left(\frac{\hat h_n-h_{n0}}{h_{n0}}\right) \to Z
\eeq
in distribution, where $Z$ is normally distributed with mean 0. The good news here is
that the relative error $(\hat h-h_{n0})/h_{n0}$ converges to 0 in
probability, as \nti. The bad news is that the rate of
convergence is very slow, $n^{-1/10}$, which confirms the large
variability of cross-validation alluded to in the introduction. 

We now explain how bagging may be applied in the c\CV~context. 
A random sample $X_1^*,\ld,X_m^*$ is drawn without
replacement from 
$X_1,\ld,X_n$, where $m<n$. This subsample is used to calculate a
least squares cross-validation bandwidth $\hat h_m$. A rescaled version of $\hat h_m$, $\tilde 
h_m=(m/n)^{1/5}\hat h_m$, is a feasible estimator of the optimal MISE
bandwidth, $h_{n0}$, for $\hat f_h$. Bagging consists of repeating
the resampling 
independently $N$ times, leading to $N$ rescaled bandwidths $\tilde
h_{m,1},\ld,\tilde h_{m,N}$. The bagging bandwidth is then defined to be
\beq\label{eq:hbag}
\hat h(m,N)=\frac{1}{N}\sum_{i=1}^N\tilde h_{m,i}.
\eeq
This approach was first proposed and studied by \cite{HR}.

It is worth mentioning that an alternative approach is to apply
bagging to the cross-validation curves, wherein one averages the
cross-validation curves from $N$ independent resamples of size $m$,
finds the minimizer of the average curve, and then rescales the
minimizer as before. The asymptotic properties of the two approaches
are equivalent, but we prefer bagging the bandwidths since doing so
requires less communication between resamples.

\section{Asymptotic results}
\label{sec:asym}
In this section, we provide asymptotic expressions for the bias and
variance of the bagging bandwidth (\ref{eq:hbag}). 
\cite{HR} studied this selector only in the case
$N=\infty$. We find that the expression they give for the variance of
(\ref{eq:hbag}), at $N=\infty$, is in error. We provide a correct
expression for this variance, and, more importantly, study the case of
finite $N$, since there is an important interplay between the values of
$m$ and $N$. Of course, in practice it is not possible to use
$N=\infty$, and indeed there is a computational motivation for
limiting the size of $N$. We will show that if $N$ is, for example, of
order $n$, then the rate of convergence of the variance to 0 is
different than in the case $N=\infty$. This is a new result that does 
not arise from the method of proof used in \cite{HR}.


Obviously $E\{ \hat h(m,N)\}= E\{(m/n)^{1/5}\hat h_m\}$, and hence it
suffices to know the bias of $(m/n)^{1/5}\hat h_m$ as an estimator 
of $h_{n0}$. We have 
\beqn
E\left\{(m/n)^{1/5}\hat h_m\right\}-h_{n0}=B_{\rm rescale}(m,n)+(m/n)^{1/5}B_{\rm CV}(m),
\eeqn
where
$$
B_{\rm rescale}(m,n)=(m/n)^{1/5}h_{m0}-h_{n0}\quad{\rm and}\quad B_{\rm CV}(m)=E(\hat
  h_m)-h_{m0}.
$$
The rescaling bias, $B_{\rm rescale}(m,n)$, is
well-understood. \cite{Marron1} shows that  
$$
B_{\rm rescale}(m,n)=\mu_{\rm rescale}m^{-2/5}n^{-1/5}+o\left(m^{-2/5}n^{-1/5}\right),
$$
where 
$$
\mu_{\rm rescale}=\frac{R(K)^{3/5}R(f''')\mu_4(K)}{20R(f'')^{8/5}}. 
$$
\cite{HR} also provide an expression for $B_{\rm rescale}(m,n)$, 
although their rate is in error.

The other bias component, $B_{\rm CV}$, is the bias inherent to
c\CV~itself, and has a curious history in the literature. In
establishing (\ref{cvrate}), \cite{HallMarron} write 
\beq
\hat h_n-h_{n0}=\xi_n+e_n,
\label{marron}
\eeq
where $E(\xi_n)=0$ and $e_n=o_p(\xi_n)$, and hence $B_{\rm CV}(n)$ is
lost in the term $e_n$. Doing so is acceptable in the case of ordinary
c\CV~because of the fact that $\var(\xi_n)$ is so large. 
In the case of bagging, however, when $\var\{\hat h(m,N)\}$ becomes sufficiently
small, one should no longer ignore $B_{\rm CV}(m)$, although this
seems to be what both \cite{Marron1} and \cite{HR} did.

In Appendix 1 as part of the proof of the main
theorem stated below, we prove that $n^{2/5}e_n$ 
converges in distribution to a random variable with mean
\beq
\mu_{\rm CV}=- \frac{8R(f)\int V(u) W(u)
  du}{25R(K)^{8/5}R(f'')^{2/5}},  
  \label{mucv}
\eeq
where $V$ and $W$ are functions determined completely by $K$. For example,
$\int V(u) W(u)du=0.1431285$ in the case of the standard normal
kernel.

The following assumptions are made in order to prove Theorem \ref{th}:
\begin{assumption}
	\label{A1} As $m,n \to \infty$, $m = o(n)$ and $N$ tends to a
        positive constant or $\infty$. 
\end{assumption}
\begin{assumption}
	\label{A2} $K$ is a symmetric and twice differentiable density
        function having, without loss of generality, variance
        $\mu_2(K)=1$. 
\end{assumption}
\begin{assumption}
	\label{A3} As $u \to \infty$, both $K(u)$ and $K'(u)$ are
        $o\{\exp(-a_1u^{a_2})\}$ for positive constants $a_1$
        and $a_2$. 
\end{assumption}
\begin{assumption}
	\label{A4} The first three derivatives of $f$ exist and are
        bounded and continuous. 
\end{assumption}


\begin{theorem}\label{th}
	Under Assumptions \ref{A1}--\ref{A4}, the bias of the bagged
        bandwidth (\ref{eq:hbag}) is: 
	\beq\label{eq:bias}
	E\left\{\hat h(m,N)\right\}-h_{n0} = m^{-1/5}n^{-1/5}\left(\mu_{\rm CV}+\mu_{\rm
		rescale}m^{-1/5}\right)+o\left(m^{-1/5}n^{-1/5}\right)
			\eeq
	and its variance is:
	\beq\label{Nvar}
	\var\left\{\hat h(m,N)\right\} &=& AC^2m^{-1/5}n^{-2/5}\left\{\frac{1}{N}+\left(\frac{m}{n}\right)^2\right\} \\ \nonumber&+& o\left(\frac{m^{-1/5}n^{-2/5}}{N}+m^{9/5}n^{-12/5}\right),
	\eeq
	where $A$ and $C$ are constants defined by (\ref{eq:Aconst}) in Appendix 1 and \eqref{eq:Cconst}, respectively. 
\end{theorem}

Expression (\ref{eq:hbag}) implies that at $N=\infty$ the asymptotic
  variance of the 
bagged bandwidth is completely determined by the covariance between
bandwidths for two different resamples. Furthermore,
to first order, as derived in \cite{BhatHart}, the 
correlation between bagged bandwidths 
from different resamples is independent of $f$ and equal to
$(m/n)^2$. This correlation is smaller when
$m$ is smaller, which is due to the fact that two resamples will usually
have fewer data values in common when $m$ is smaller.  
In fact, taking $N=\infty$ yields the approximation 
\begin{eqnarray}\label{Ninf}
\var\left\{\hat h(m,N)\right\}=AC^2m^{9/5}n^{-12/5}+o\left(m^{9/5}n^{-12/5}\right),
\end{eqnarray}
which matches precisely one of the two summands in expression (13) of
\cite{HR}. It can be shown that the other summand, rather than being
the dominant term, as claimed by \cite{HR},  is actually
negligible in comparison to (\ref{Ninf}).

It is easily verified that the choice of $m$ that
minimizes the main term of (\ref{Nvar}) is asymptotic to $n/(3\sqrt
N)$. Therefore, if 
$N=n$, say, then the fastest rate at which $\var\{\hat
h(m,N)\}/h_{n0}^2$ can converge to 0 is $n^{-11/10}$. In contrast, when
$N=\infty$, the rate of convergence of $\var\{\hat
h(m,N)\}/h_{n0}^2$ can be arbitrarily close to $n^{-2}$ by allowing $m$ to
increase sufficiently slowly with $n$. This makes it clear that the
properties of the bagged bandwidth are substantially affected by how many
subsamples are taken, and hence it does not suffice to analyze
the bagged bandwidth by setting $N=\infty$.

It is remarkable how much stability bagging can provide. Whether $N$ is
$\infty$ or merely tending to $\infty$, $\var\{\hat h(m,N)/h_{n0}\}$ can
converge to 0 faster than the usual parametric rate of $n^{-1}$. This
is in stark contrast to the extremely slow rate of $n^{-1/5}$ for
ordinary c\CV. Unfortunately, this extreme stability cannot be fully
taken advantage of since the bagged bandwidth is more biased than the
ordinary c\CV~bandwidth. The largest reductions in variance are 
associated with small values of $m$, but it turns out that small $m$
yields the largest bias. 

As seen in \eqref{eq:bias}, the bias term $B_{\rm CV}$,
that has been ignored to date, is of a larger order than the
rescaling bias. This and the fact that $\mu_{\rm CV}<0$ suggest that
the bagged bandwidth would tend to be smaller than the optimal
bandwidth $h_{n0}$. However, our experience in numerous simulations is
that the 
bagged bandwidth actually tends to be larger than $h_{n0}$. The
explanation for this phenomenon is simple: $\mu_{\rm rescale}>0$ and
$\mu_{\rm rescale}$ is larger than $|\mu_{CV}|$ in every case we have
checked. Indeed, we have not found a case where  $\mu_{\rm
  rescale}/|\mu_{CV}|$ is less than 2, and it appears that there is no 
limit to how large this ratio can be.  

Table \ref{tab1} provides the constants $\mu_{\rm
  rescale}$ and $\mu_{CV}$ for several densities. Two patterns are
apparent here: (i) the heavier the tail of the density, the
more dominant is the rescaling bias, and (ii) the rescaling bias is more
dominant for multimodal mixtures of normals than for the normal 
itself. Define $m_{\rm crit}$ to be the smallest subsample size at
which the asymptotic mean of the bagged bandwidth is not larger than
the optimal MISE bandwidth, $h_{n0}$. Since the ratio $\mu_{\rm
  rescale}/\mu_{CV}$ is invariant to location and scale, it
follows that the values of $m_{\rm crit}$ for any normal, logistic or
Cauchy 
distribution are the same as in Table \ref{tab1}. Except in the case
of the Beta$(5,5)$ and normal densities, the values of $m_{\rm crit}$
are very large, especially considering
that a good choice for $m$ is usually much smaller than $n$, as we shall subsequently see. So, in
spite of what the asymptotics suggest, it will often be the case that
the bagged bandwidth is larger on average than the optimal
bandwidth. This is a classic case of asymptotics not kicking in
until the sample size is extremely large.

\begin{table}[h]
\begin{center}
\begin{threeparttable}
	\def~{\hphantom{0}}
	\caption{Bias constants and critical $m$ ($m_{\rm crit}$) for the Gaussian
		kernel. The claw density \citep{MaWa} is a
		symmetric density mixture of six normals with five modes (see Section \ref{sec:sim}, for the definition of the claw density). The bimodal mixture of two normals has parameters $\mu = (-1.5, 1.5)$, $\sigma = (0.5, 0.5)$ and $w = (0.5, 0.5)$, where $\mu$, $\sigma$ and $w$ are the mean, standard deviation and weight vectors, respectively, for the density mixture (see Section \ref{sec:sim}, for the notation used for a normal mixture density).	\label{tab1}}
		\begin{tabular}{lccc}
		\hline
			Density & $\mu_{\rm rescale}$ & $\mu_{\rm CV}$ & $m_{\rm crit}$ \\
			\hline
			Beta$(5,5)$ & $0.06554$ & $-0.03070$ & $45$ \\
			Standard normal & $0.44565$ & $-0.18216$ & $88$ \\
			Standard logistic & $0.92556$ & $-0.25787$ & $596$ \\
			Bimodal mixture of two normals &$0.32809$ & $-0.05988$ & $4936$\\
			Standard Cauchy & $1.24349$ & $-0.09793$ & $330,154$ \\
			Claw & $0.22774$ & $-0.00766$ & $>10^7$\\
			\hline
		\end{tabular}
		\begin{tablenotes}
	\item The bimodal mixture of two normals has parameters $\mu = (-1.5, 1.5)$, $\sigma = (0.5, 0.5)$ and $w = (0.5, 0.5)$, where $\mu$, $\sigma$ and $w$ are the mean, standard deviation and weight vectors, respectively, for the density mixture (see Section \ref{sec:sim}, for the notation used for a normal mixture density).
	\end{tablenotes}
 \end{threeparttable}
 \end{center}
\end{table}


\section{Choosing an optimal subsample size}

\label{sec:choos}
In practice, for fixed $n$ and $N$, our results allow one to
  estimate an optimal subsample size, $m_0$. This quantity
  is defined to be the minimizer of the asymptotic mean squared
  error (\amse) of $\hat{h}(m,N)$ with respect to $m$: 
\beq
\amse \left\{\hat{h}(m,N)\right\} &=& AC^2m^{-1/5}n^{-2/5}\left\{\frac{1}{N}+\left(\frac{m}{n}\right)^2\right\}\nonumber\\ &+& m^{-2/5}n^{-2/5}\left(\mu_{CV}+\mu_{rescale}m^{-1/5}\right)^2.
\label{amse}
\eeq

Since $\mu_{rescale}$, $\mu_{CV}$, $A$ and $C$ are unknown, we propose
the following method to estimate $$m_0 =
\argmin\limits_{m>1}\amse\left\{\hat{h}(m,N)\right\}.$$ 

\begin{enumerate}
\item Consider $s$ subsamples of size $r < n$, drawn without replacement from the original sample of size $n$.
\item For each of these subsamples, fit a normal mixture model. 
To fit a mixture model with a given number of components, use the expectation-maximization algorithm initialized by hierarchical model-based agglomerative clustering.  Then, estimate the optimal number of mixture components by using BIC, the Bayesian information criterion. 
In practice, this process is performed employing the R package \texttt{mclust} \citep[see][]{mclust}.
\item Use $R(\hat{f}_i)$, $R(\hat{f}_i'')$ and $R(\hat{f}_i''')$ to estimate $A$, $C$, $\mu_{CV}$ and $\mu_{rescale}$, where $\hat{f}_i$ denotes the density function of the normal mixture fitted to the $i$th subsample. Denote these estimates by $\hat A_i$, $\hat C_i$, $\hat \mu_{CV, i}$ and $\hat \mu_{rescale, i}$.
\item Compute the bagged estimates of the unknown constants,
  that is, $\hat D = \frac{1}{s}\sum\limits_{i=1}^s \hat D_i$, 
where $\hat D_i$ can be $\hat A_i$, $\hat C_i$, $\hat \mu_{CV, i}$ or
$\hat \mu_{rescale, i}$, and obtain $\widehat{\amse}\{\hat{h}(m,N)\}$ by
plugging these bagged estimates into (\ref{amse}). 
\item Finally, estimate $m_0$ by: $$\hat m_0 =
  \argmin\limits_{m>1}\widehat{\amse}\left\{\hat{h}(m,N)\right\}.$$ 
\end{enumerate}

Regarding the selection of $s$ and $r$ in Step 1, we have performed
some empirical tests and observed that the estimation of $h_{n0}$ by
$\hat h(\hat m_0, N)$ is quite robust to the values of these
parameters. For example, values of $s \simeq  50$  and $r \simeq 0.01
n$ have provided, in general, good results. 

\section{Discussion}
\label{sec:concl}

The finite sample behaviour of a bagged
cross-validation bandwidth was investigated by means of a
simulation study, and its practical performance was
illustrated using a large data set involving flight delays. These experiments, included in Appendixes 2 and 3, show that subsampling can significantly reduce
computing time relative to a binned version of leave-one-out
cross-validation.

As mentioned in Section \ref{intro}, bagged versions
of other bandwidth selection methodologies, such as plug-in and
bootstrap, could be considered. While both cross-validation and
bootstrap approaches try to estimate $h_{n0}$, plug-in bandwidths are
estimators of $h_{na}$, the bandwidth minimizing the asymptotic
  MISE, and hence they only need to estimate $R(f'')$. It is worth
noting that there is a clear similarity between the three
methods. Both cross-validation \citep{scottterrell} and bootstrap
\citep{cao1993} bandwidths are minimizers of criteria of the form:
\begin{equation}\label{eq:criterion}
\sum\limits_{(i,j) \in \cal{I}} H_{nhg}(X_i-X_j)+\frac{R(K)}{nh},
\end{equation}
where ${\cal I} \subset \{1,\ldots,n\} \times \{1,\ldots,n\}$ and
$H_{nhg}$ is a function that may depend on the sample size, $n$, the
bandwidth, $h$, and a pilot bandwidth, $g$. Note that $g$ plays a role
only in the bootstrap criterion. Although plug-in bandwidths are not
solutions to a minimization problem, the nonparametric estimation of
$R(f'')$ using pilot bandwidth $g$ requires working with a
$U$-statistic like the one given in the first term in
\eqref{eq:criterion}, which would only depend on $n$ and $g$. Due to
the nonlinearity of \eqref{eq:criterion} with respect to the
observations, it stands to reason that a bagged implementation of
these methods could reduce their variability, as in the case of
cross-validation.

 \section*{Acknowledgements}\label{acknowledgements}
The authors thank Andrew Robinson, an anonymous referee,
  the Editor and an Associate Editor for numerous useful
  comments that significantly improved this article. The authors are
  also grateful for the insight of Professor Anirban Bhattacharya, who worked
  with Professor Hart on partitioned cross-validation, a method
  closely related to bagged c\CV. 

This research has been supported by MINECO Grant MTM2017-82724-R, and
by the Xunta de Galicia (Grupos de Referencia Competitiva
ED431C-2016-015 and ED431C-2020-14, and Centro Singular de Investigaci\'on de Galicia
ED431G 2019/01), for the first three authors, all of them through the
ERDF. Additionally, the work of the first author was carried out
during a visit at Texas A\&M University, College Station, financed by
INDITEX, with reference INDITEX-UDC 2019.

\vspace*{-10pt}

\appendix

\section*{Appendix 1. Theoretical results} \label{sec:teo}

This Appendix includes the proof of Theorem \ref{th}, providing the asymptotic bias and variance of our bagged cross-validation bandwidth.

To prove Theorem \ref{th}, we establish one lemma in advance. 
\begin{lema} \label{lemma1}
	Under Assumptions \ref{A1}--\ref{A4},
	\beq
	n^{1/5}CV'''(\tilde h_n) = o_P(1),
	\label{conj}
	\eeq
	where $\tilde h_n$ is a bandwidth between the cross-validation bandwidth 	$\hat h_n$ and the {\rm MISE} minimizer $h_{n0}$.
\end{lema}

\begin{proof}
	First, we write  
	\beq
	n^{1/5}CV'''(\tilde h_n) = \alpha_1 + \alpha_2,
	\label{alphas}
	\eeq
	with $\alpha_1 = n^{1/5}CV'''(h_{n0})$ and $\alpha_2 = n^{1/5}\left\{
	CV'''(\tilde h_n) -CV'''(h_{n0}) \right\}$. To prove Lemma 1 it is
	sufficient to show that $\alpha_1 = o_P(1)$ and $ \alpha_2 =
	o_P(1)$. In order to study the term $\alpha_1$, we first consider the
	asymptotic MISE of the Parzen--Rosenblatt estimator of the density
	function. It is well-known that if
	$K$ is a second order symmetric kernel function and considering
	that $K$ has variance $1$, as stated in Assumption \ref{A2}, the MISE
	is: 
	\beqn
	M(h) = \frac{R(K)}{nh}+\frac{1}{4}h^4R(f'')+o\left\{(nh)^{-1}+h^4\right\},
	\eeqn
	and, hence,
	\beqn
	M'''(h) = -\frac{6R(K)}{nh^4}+6hR(f'')+o\left\{(nh^4)^{-1}+h\right\}.
	\eeqn
	Since 
	\beqn
	-\frac{6R(K)}{nh_{na}^4}+6h_{na}R(f'') = 0,
	\eeqn
	where $h_{na}$ denotes the bandwidth minimizing the asymptotic MISE, it follows immediately that $n^{1/5}M'''(h_{n0})$ converges to 0.
	Now, we can write
	\beqn
	n^{1/5}CV'''(h_{n0}) = n^{1/5}M'''(h_{n0})+n^{1/5}\eta_n,
	\eeqn
	where $\eta_n = CV'''(h_{n0})-M'''(h_{n0})$. Thus, to prove that
	$\alpha_1 = o_P(1)$, it is sufficient to prove that  
	\beq
	\eta_n = o_P\left(n^{-1/5}\right),
	\label{epsilon}
	\eeq
	or, by Markov's inequality, that
	$n^{2/5}\var\{CV'''(h_{n0})\}=o(1)$.
	It is easy to prove that, for every $r \geq 1$,
	\beq
	CV^{(r)}(h) = M^{(r)}(h)+\frac{1}{n(n-1)}\sum_{i \neq j} \bar{\gamma}_{nh}^{(r)}(X_i-X_j),
	\label{cvr}
	\eeq
	where 
	$$\gamma_n(u) = \frac{n-1}{n}K*K(u)-2K(u),$$ 
	$$\gamma_{nh}(u) = \gamma_n(u/h)/h,$$ 
	$$\bar{\gamma}_{nh}(u) = \gamma_{nh}(u)-E\{\gamma_{nh}(X_1-X_2)\}$$ 
	and
	$$\bar{\gamma}_{nh}^{(r)}(u) = \frac{{\rm d}^r \bar{\gamma}_{nh}(u)}{{\rm d}h^r}.$$
	 Therefore,
	\beqn
	\var\left\{CV'''(h)\right\} = \frac{1}{n^2(n-1)^2}\sum_{\substack{i,j,k,l=1\\i \neq j\\ k \neq l}}^n \cov\left\{\Psi_3(X_i-X_j),\Psi_3(X_k-X_l)\right\},
	\eeqn
	where
	\beqn
	\Psi_3(u) = \frac{{\rm d}^3 \gamma_{nh}(u)}{{\rm d}h^3} = -\left\{\frac{6}{h^4}\gamma_n(u/h)+\frac{18u}{h^5}\gamma_n'(u/h)+\frac{9u^2}{h^6}\gamma_n''(u/h)+\frac{u^3}{h^7}\gamma_n'''(u/h)\right\}.
	\eeqn
	
	Counting the different possible cases, we get
	\beq
	\var\left\{CV'''(h)\right\}  & = &  \frac{1}{n^2(n-1)^2}\left[4n(n-1)(n-2)\cov\left\{\Psi_3(X_1-X_2),\Psi_3(X_1-X_3)\right\} \right. \nonumber \\  & + & \left.2n(n-1)\var\left\{\Psi_3(X_1-X_2)\right\}\right]. \nonumber
	\eeq
	
	Let us now define the function $\tilde \Psi_3(u)$, such that, $\Psi_3(u) = \tilde \Psi_3(u/h)/h$. Consequently,
	\beqn
	\tilde \Psi_3(u) = -\frac{1}{h^3}\left\{6\gamma_n(u)+18u\gamma_n'(u)+9u^2\gamma_n''(u)+u^3\gamma_n'''(u)\right\}.
	\eeqn
	
	Taking into account the definition of $\mu_j(g) = \int x^jg(x)\,
	{\rm d}x$, $j=0,1, \ldots$, for any function $g$, 
	we shall now proceed to compute $\mu_j\left(\tilde \Psi_3\right)$, for $j = 0,2,4,6$, and $\mu_j\left(\tilde \Psi_3^2\right)$, for $j = 0,2$, since we will need these quantities later on. Note that $\mu_j\left(\tilde \Psi_3\right) = 0$, for every odd $j$, since $\tilde \Psi_3$ is symmetric.
	
	For $j=0$,
	\beqn
	\mu_0(\tilde{\Psi}_3) = -\frac{1}{h^3}\left\{6\mu_0(\gamma_n)+18\mu_1(\gamma_n')+9\mu_2(\gamma_n'')+\mu_3(\gamma_n''')\right\}.
	\eeqn
	
	Using integration by parts and the fact that $\mu_0(K) = \mu_0(K*K) =
	1$, we get 
	\beqn
	\mu_0(\gamma_n) &=& -\frac{n+1}{n},\\
	\mu_1(\gamma_n') &=& \frac{n+1}{n},\\
	\mu_2(\gamma_n'') &=& -2\left(\frac{n+1}{n}\right),\\
	\mu_3(\gamma_n''') &=& 6\left(\frac{n+1}{n}\right),
	\eeqn
	and, hence,
	\beqn
	\mu_0(\tilde{\Psi}_3) = 0.
	\eeqn
	
	Now,
	\beqn
	\mu_2(\tilde{\Psi}_3) = -\frac{1}{h^3}\left\{6\mu_2(\gamma_n)+18\mu_3(\gamma_n')+9\mu_4(\gamma_n'')+\mu_5(\gamma_n''')\right\}.
	\eeqn
	Partial integration and the equality $\mu_2(K*K) =
	2\mu_2(K)$ give 
	\beqn
	\mu_2(\gamma_n) &=& -2\mu_2(K)/n,\\
	\mu_3(\gamma_n') &=& 6\mu_2(K)/n,\\
	\mu_4(\gamma_n'') &=& -24\mu_2(K)/n,\\
	\mu_5(\gamma_n''') &=& 120\mu_2(K)/n,
	\eeqn
	and, therefore,
	\beqn
	\mu_2(\tilde{\Psi}_3) = 0.
	\eeqn
	
	We have
	\beqn
	\mu_4(\tilde{\Psi}_3) = -\frac{1}{h^3}\left\{6\mu_4(\gamma_n)+18\mu_5(\gamma_n')+9\mu_6(\gamma_n'')+\mu_7(\gamma_n''')\right\}.
	\eeqn
	Using integration by parts and the fact that $\mu_4(K*K) =
	2\mu_4(K)+6\mu_2(K)^2$, we get 
	\beqn
	\mu_4(\gamma_n) &=& 6\mu_2(K)^2-2\mu_4(K)/n,\\ 
	\mu_5(\gamma_n') &=& -30\mu_2(K)^2+10\mu_4(K)/n,\\ 
	\mu_6(\gamma_n'') &=& 180\mu_2(K)^2-60\mu_4(K)/n, \\
	\mu_7(\gamma_n''') &=& -1260\mu_2(K)^2+420\mu_4(K)/n,
	\eeqn
	and, therefore,
	\beqn
	\mu_4(\tilde{\Psi}_3) = \frac{144\mu_2(K)^2}{h^3}+O\left(\frac{1}{nh^3}\right).
	\eeqn
	
	Finally,
	\beqn
	\mu_6(\tilde{\Psi}_3) = -\frac{1}{h^3}\left\{6\mu_6(\gamma_n)+18\mu_7(\gamma_n')+9\mu_8(\gamma_n'')+\mu_9(\gamma_n''')\right\}.
	\eeqn
	Using integration by parts and the fact that $\mu_6(K*K) =
	2\mu_6(K)+30\mu_2(K)\mu_4(K)$, we get 
	\beqn
	\mu_6(\gamma_n) &=& 30\mu_2(K)\mu_4(K) + O(1/n),\\ 
	\mu_7(\gamma_n') &=& -210\mu_2(K)\mu_4(K) + O(1/n),\\ 
	\mu_8(\gamma_n'') &=& 1680\mu_2(K)\mu_4(K) + O(1/n),\\
	\mu_9(\gamma_n''') &=& -15120\mu_2(K)\mu_4(K) + O(1/n),
	\eeqn
	and so
	\beqn
	\mu_6(\tilde{\Psi}_3) =
	\frac{3600\mu_2(K)\mu_4(K)}{h^3}+O\left(\frac{1}{nh^3}\right). 
	\eeqn
	
	Analogously, it can be proved that
	\beqn
	\mu_0(\tilde{\Psi}_3^2) = \mu_2(\tilde{\Psi}_3^2) = O\left(\frac{1}{h^6}\right).
	\eeqn
	On the other hand,
	\beqn
	\var\left\{\Psi_3(X_1-X_2)\right\}= I_1-I_2^2
	\eeqn
	and
	\beqn
	\cov\left\{\Psi_3(X_1-X_2),\Psi_3(X_1-X_3)\right\} = I_3-I_2^2,
	\eeqn
	where
	\beqn
	I_1 &=& \int \Psi_3^2*f(x)f(x)\,{\rm d}x,\\
	I_2 &=& \int \Psi_3*f(x)f(x)\,{\rm d}x,\\
	I_3 &=& \int \Psi_3*f(x)^2f(x)\,{\rm d}x.
	\eeqn
	
	Simple algebra and Taylor expansions give
	\beqn
	I_1 &=& \frac{1}{h}\int\int \tilde{\Psi}_3(u)^2f(x)\left\{f(x)+\frac{h^2u^2}{2}f''(\zeta)\right\}\,{\rm d}x\, {\rm d}u \\&=& \frac{1}{h}\left[\mu_0(\tilde{\Psi}_3^2)R(f)+O\left\{h^2\mu_2(\tilde{\Psi}_3^2)\right\}\right] = O\left(\frac{1}{h^7}\right),
	\eeqn
	\beqn
	I_2 &=& \int\int \tilde{\Psi}_3(u)f(x)\left\{\frac{h^4u^4}{4!}f^{(4)}(x)+\frac{h^6u^6}{6!}f^{(6)}(\xi)\right\}\,{\rm d}x \, {\rm d}u \\&=& \frac{h^4}{24}\mu_4(\tilde{\Psi}_3)R(f'')+O\left\{h^6\mu_6(\tilde{\Psi}_3)\right\} = 6\mu_2(K)^2R(f'')h+O\left(h^3\right),
	\eeqn
	and
	\beqn
	I_3 &=& \int f(x)\left\{\int\frac{1}{h}\tilde{\Psi}_3\left(\frac{x-y}{h}\right)f(y)\,{\rm d}y\right\}^2\,{\rm d}x \\&=& \int f(x)\left\{6\mu_2(K)^2f^{(4)}(x)h+O\left(h^3\right)\right\}^2\,{\rm d}x \\&=& 36\mu_2(K)^4\int f^{(4)}(x)^2f(x)\,{\rm d}x h^2 + O\left(h^4\right).
	\eeqn
	Therefore,
	\begin{eqnarray*}
		\var\left\{\Psi_3(X_1-X_2)\right\} &=& O\left(\frac{1}{h^7}\right),\\
		\cov\left\{\Psi_3(X_1-X_2),\Psi_3(X_1-X_3)\right\} &=& \mathcal{L}h^2+O\left(h^4\right),
	\end{eqnarray*}
	where $\mathcal{L} = 36\mu_2(K)^4\left\{\int f^{(4)}(x)^2f(x)\,{\rm d}x -
	R(f'')^2\right\}$. Consequently, 
	\begin{eqnarray}
	\var\left\{CV'''(h)\right\} = O\left(\frac{1}{n^2h^7}\right),
	\label{varCV'''}
	\end{eqnarray}
	and $\var\left\{CV'''(h_{n0})\right\} = O\left(n^{-3/5}\right)$. 
	Therefore, as required, $\var\left\{CV'''(h_{n0})\right\} =
	o\left(n^{-2/5}\right)$ and so $\alpha_1 = o_P(1)$.
	
	To handle the term $\alpha_2$ in (\ref{alphas}), we write
	\beq
	\alpha_2 = n^{1/5}\left\{ CV'''(\tilde h_n) - CV'''(h_{n0}) \right\} = n^{1/5} (\tilde h_n - h_{n0}) CV^{(4)}(\overline h_n),
	\label{alpha_2} 
	\eeq
	where $\overline h_n$ is an intermediate value between $\tilde h_n$
	and $h_{n0}$. The results of \cite{HallMarron} imply that   
	$\tilde h_n-h_{n0} = O_P\left(n^{-3/10}\right)$. Thus,  
	in view of (\ref{alpha_2}), to prove $\alpha_2 = o_P(1)$ it is sufficient
	to show that   
	\beq
	n^{-1/10} \sup_{h \in I(h_n, h_{n0})} |CV^{(4)}(h)| = o_P(1),
	\label{supCV(4)} 
	\eeq
	where $I(h_n, h_{n0})$ is the interval with endpoints $h_n$ and $h_{n0}$.
	
	Let $a$ be arbitrarily small but fixed, and such that
	$an^{-1/5}<h_{n0}<a^{-1}n^{-1/5}$. Without loss of generality, we
	suppose that $CV(h)$ is minimized over a finite set $I_n$ having 
	equally spaced points on the interval $(an^{-1/5},a^{-1}n^{-1/5})$. It is
	assumed that the number of points in $I_n$ is $n^{2/5-d}$, where
	$0<d<1/5$. Let $h_n^*$ be the minimizer of 
	$M(h)$ over $I_n$. Then optimizing $CV$ over $I_n$ suffices since
	$h_n^*-h_{n0}$ is of order   $n^{-3/5+d}$, implying that this
	source of error is smaller than $n^{-2/5}$ and hence negligible for
	the current argument. It is 
	enough to show that $n^{-1/10}\max_{h\in I_n}|CV^{(4)}(h)|$ converges in
	probability to 0. Since
	$|CV^{(4)}(h)|\le|CV^{(4)}(h)-E_n(h)|+|E_n(h)|$, where $E_n(h)=E\left\{
	CV^{(4)}(h) \right\}$, it suffices to show that
	$\lim_{n\ra\infty}n^{-1/10}\max_{h\in I_n}|E_n(h)|=0$ and $n^{-1/10}\max_{h\in
		I_n}|CV^{(4)}(h)-E_n(h)|=o_P(1)$.   
	
	For any $\epsilon>0$, we have
	\begin{eqnarray*}
		P\left\{n^{-1/10}\max_{h\in
			I_n}|CV^{(4)}(h)-E_n(h)|\ge\epsilon\right\}&\le& P\left[\bigcup_{h\in
			I_n}\left\{n^{-1/10}|CV^{(4)}(h)-E_n(h)|\ge\epsilon\right\}\right]\\
		&\le&\sum_{h\in
			I_n}P\left\{n^{-1/10}|CV^{(4)}(h)-E_n(h)|\ge\epsilon\right\}\\
		&\le& \sum_{h\in I_n}\frac{\Var\left\{CV^{(4)}(h)\right\}}{n^{1/5}\epsilon^2}\\
		&\le&\frac{n^{1/5-d}}{\epsilon^2}\,\max_{h\in I_n}\Var\left\{CV^{(4)}(h)\right\}.
	\end{eqnarray*}

	Let us now obtain uniform bounds for the expectation and
	variance of $CV^{(4)}(h)$. It is straightforward to prove that  
	\beqn
	E_n(h) = M^{(4)}(h) \sim 6 \mu_2(K)^2 R(f'')  + 24 R(K) n^{-1} h^{-5}
	\eeqn
	and, since $h_{n0} \sim h_{na} = C n^{-1/5}$, we have that $E_n(h_{n0}) \sim \mathcal{D}$, for some constant $\mathcal{D}>0$. 
	On the other hand, since $I_n \subset [an^{-1/5},a^{-1}n^{-1/5}]$, we get
	\beq
	\max_{h \in I_n} E \left\{ CV^{(4)}(h) \right\} = O(1).
	\label{supECV(4)} 
	\eeq
	To obtain a uniform bound for the variance, long and tedious
	calculations can be performed to get a similar  
	expression to (\ref{varCV'''}), but for the fourth derivative:
	\begin{eqnarray*}
		\var\left\{CV^{(4)}(h)\right\} = O\left(\frac{1}{n^2h^9}\right).
	\end{eqnarray*}
	Using again $h_{n0} \sim C n^{-1/5}$ and $I_n \subset [an^{-1/5},a^{-1}n^{-1/5}]$, we obtain
	\beq
	\max_{h \in I_n} \var \left\{ CV^{(4)}(h) \right\} = O(n^{-1/5}).
	\label{supVarCV(4)} 
	\eeq
	
	Using expressions (\ref{supECV(4)}) and (\ref{supVarCV(4)}), it now
	follows that  
	$$\max_{h \in I_n}n^{-1/10}|CV^{(4)}(h)| = o_P(1),$$ thus completing the
	proof.
\end{proof}

\begin{proof}[of Theorem~\ref{th}]
	The variance of the bagging bandwidth is:
	\begin{eqnarray}\label{bagvar}
	\var\left\{\hat h(m,N)\right\}=\frac{1}{N}\var\left(\tilde h_{m,1}\right)+\frac{N-1}{N}\cov\left(\tilde h_{m,1},\tilde h_{m,2}\right). 
	\end{eqnarray}
	
	The work of \cite{HallMarron} provides an approximation to
	the variance of $\tilde h_{m,1}$: 
	\begin{eqnarray}\label{Var}
	\frac{\Var\left(\tilde h_{m,1}\right)}{h_{n0}^2}=Am^{-1/5}+o\left(m^{-1/5}\right),
	\end{eqnarray}
	where
	\beq\label{eq:Aconst}
	A = \frac{8R(V)R(f)\mu_2(K)^{4/5}}{25R(K)^{9/5}R(f'')},
	\eeq
	the function $V$ is defined in \cite{BhatHart} and only depends on the kernel $K$ and $\mu_j(g) = \int x^j g(x)\,{\rm d}x$ for $j= 0,1,2,\dots$ \cite{BhatHart} derive the following approximation to the last term in
	(\ref{bagvar}): 
	\begin{eqnarray}\label{Cov}
	\cov\left\{\tilde h_{m,1},\tilde h_{m,2}\right\}=\Var\left(\tilde
	h_{m,1}\right)\left(\frac{m}{n}\right)^2+o\left(m^{9/5}n^{-12/5}\right). 
	\end{eqnarray}
	Plugging (\ref{Var}) and (\ref{Cov}) into (\ref{bagvar}),
	when $N$ is either fixed or tending to $\infty$ with $n$, then, 
	\beqn
	\var\left\{\hat h(m,N)\right\}\sim
	AC^2m^{-1/5}n^{-2/5}\left\{\frac{1}{N}+\left(\frac{m}{n}\right)^2\right\}.
	\eeqn
	
	Regarding the bias of $\hat h(m,N)$, as explained in Section 3 of the main paper, we only have to focus on deriving the bias inherent to cross-validation itself. Let $\hat h_n$ be the ordinary cross-validation bandwidth for a sample of size $n$, and let $h_{n0}$ be the minimizer of MISE, $M(h)$. Using the fact that $CV'(\hat h_n) = 0$, a Taylor expansion gives
	\beqn
	\hat h_n-h_{n0} = -\frac{CV'(h_{n0})}{CV''(h_n)}
	\eeqn
	for $h_n$ between $\hat h_n$ and $h_{n0}$. Now expand $1/CV''(h_n)$ in a Taylor series about $\Delta = M''(h_{n0})$, yielding
	\beqn
	\hat h_n-h_{n0} = -\frac{CV'(h_{n0})}{\Delta}+\frac{CV'(h_{n0})\left\{CV''(h_n)-\Delta\right\}}{\hat{\Delta}^2},
	\eeqn
	where $\hat \Delta$ is between $CV''(h_n)$ and $M''(h_{n0})$.
	
	Using the notation in  equation (4) of the main paper,
        $\xi_n = -CV'(h_{n0})/\Delta$ and  
	$$
	e_n=\frac{CV'(h_{n0})\left\{CV''(h_n)-\Delta\right\}}{\hat \Delta^2}.
	$$
	
	The random variable $-CV'(h_{n0})/\Delta$ has mean $0$ and is $O_P\left(n^{-3/10}\right)$, as shown by \cite{HallMarron}. We will show that $n^{2/5}e_n \to Y$ in distribution, where $E(Y) = \mu_{CV} < 0$ and $\var(Y) > 0$, with $\mu_{CV}$ as in equation (5) in the main paper. In effect, this will establish the first order bias of $\hat h_n$ as an estimator of $h_{n0}$. Using results of \cite{HallMarron}, $n^{4/5}\hat \Delta^2 \to D^2 > 0$ in probability, where $D$ is the limit of $n^{2/5}M''(h_{n0})$ as $n \to \infty$. It is sufficient then to consider
	\beq
	n^{6/5}CV'(h_{n0})\left\{CV''(h_n)-\Delta\right\} = n^{6/5}CV'(h_{n0})\left\{CV''(h_{n0})-\Delta+\delta_n\right\},
	\label{nume}
	\eeq
	where $\delta_n = CV''(h_n)-CV''(h_{n0})$. Now,
	\beqn
	\delta_n = (h_n-h_{n0})CV'''(\tilde h_n),
	\eeqn
	where $\tilde h_n$ is between $h_n$ and $h_{n0}$. From \cite{HallMarron}, we know that $CV'(h_{n0}) = O_P\left(n^{-7/10}\right)$ and $h_n-h_{n0} = O_P\left(n^{-3/10}\right)$. It follows that
	\beqn
	n^{6/5}CV'(h_{n0})\delta_n = CV'''(\tilde h_n) O_P(n^{1/5}).
	\eeqn
	
	Considering Lemma \ref{lemma1}, in equation (\ref{nume}), we need only investigate
	\beqn
	n^{6/5}CV'(h_{n0})\left\{CV''(h_{n0})-\Delta\right\}.
	\eeqn
	
	\cite{HallMarron} show that
	\beqn
	n^{7/10}CV'(h_{n0}) \to N(0,\sigma_1^2)
	\eeqn
	in distribution. As shown in \cite{BhatHart}, $h_{n0}\left\{CV''(h_{n0})-\Delta\right\}$ is identical in structure to $CV'(h_{n0})$ and, hence,
	\beqn
	n^{7/10}h_{n0}\left\{CV''(h_{n0})-\Delta\right\} \sim C_0 \sqrt{n}\left\{CV''(h_{n0})-\Delta\right\} \to N(0, \sigma_2^2)
	\eeqn
	in distribution. Using the Cram\'er-Wold device, it follows that
	\beqn
	\sqrt{n}\left\{n^{1/5}CV'(h_{n0}), CV''(h_{n0})-\Delta\right\}
	\eeqn
	converges in distribution to a bivariate normal random variable with mean vector \textbf{$0$} and covariance matrix \textbf{$\Sigma$}. Using Theorem B., p. 124 of \cite{serfling}, we have
	\beqn
	n^{6/5}CV'(h_{n0})\left\{CV''(h_{n0})-\Delta\right\} \to Y_1Y_2
	\eeqn
	in distribution, where $(Y_1, Y_2)$ are bivariate normal with mean vector \textbf{$0$} and covariance matrix \textbf{$\Sigma$}. 
	\cite{BhatHart} show that $E(Y_1Y_2)$ is 
	\beqn
	-\frac{8}{R(K)^{4/5} R(f'')^{-4/5}}\int V(u) W(u) \, {\rm d}u \int f^2(x)\,{\rm d}x.
	\eeqn
	
	Also, taking into account that \citep{BhatHart}
	$$
	M''(h_{n,0}) \sim 5 R(K)^{2/5} R(f'')^{3/5} n^{-2/5},
	$$
	the limiting expectation of $n^{2/5}(\hat h_n-h_{n0})$ is
	$$
	\frac{E(Y_1Y_2)}{D^2} = \mu_{CV}= - \frac{8R(f)\int V(u) W(u)
		du}{25R(K)^{8/5}R(f'')^{2/5}},  
	$$
	which completes the proof.
\end{proof}

\section*{Appendix 2. Simulation study}
\label{sec:sim}
To test the behaviour of the bagged cross-validation bandwidth (3) in the main paper, some simulation studies were performed considering different density functions, sample sizes ($n$), subsample sizes ($m$), and number of subsamples ($N$). For the sake of brevity, we only present the results obtained for two normal mixture densities, although similar results were obtained for other densities.  We denote by $\mu=(\mu_1, \ldots, \mu_k)$, $\sigma=(\sigma_1, \ldots, \sigma_k)$ and $w=(w_1, \ldots, w_k)$ the mean, standard deviation and weight vectors, respectively, for the density mixture $f\left(x\right)=\sum_{i=1}^{k}w_{i}\phi_{\mu_{i},\sigma_{i}}$, with $\phi_{\mu_{i},\sigma_{i}}$ a $N\left(\mu_{i},\sigma_{i}\right)$
density, $i=1, \ldots, k$. Here, we consider the density mixture of two normals, denoted by D1, with parameters $\mu = (0, 1.5)$, $\sigma = (1, 1/3)$ and $w = (0.75, 0.25)$, and the claw density, denoted by D2, mixture of six normals, with parameters $\mu = (0,-1,-0.5,0,0.5,1)$, $\sigma = (1,0.1,0.1,0.1,0.1,0.1)$ and $w = (0.5,0.1,0.1,0.1,0.1,0.1)$.

In this experiment, $1,000$ samples of size $n = 10^5$ were simulated from the previous densities and the bagged, $\hat{h}(m,N)$, and leave-one-out cross-validation, $\hat{h}_n$, bandwidths were computed. The bagged bandwidths were calculated using $N=500$ subsamples and considering four values for the size of the subsamples, $m$, including the theoretical optimal values, $m_0 = 13,081$ and $m_0 = 20,326$, for densities D1 and D2, respectively. For each sample, we also computed the estimated $m_0$ using the algorithm presented in Section 4 of the main paper, with values $s=50$ and $r \in \{500, 1,000, 5,000\}$ in Step 1. 
The Gaussian kernel was used throughout the study. 
The R \citep{Rsoft} package \texttt{baggedcv} \citep{baggedcv}
was employed to carry out the simulation experiments.

To compute the different cross-validation bandwidths involved in this simulation, $\hat{h}_n$ and $\hat{h}_{m,i}$, $i=1, \ldots, N$, we employed the R function \texttt{bw.ucv}. 
This function uses a binned implementation and, therefore, it is extremely fast. However, when the number of bins, \texttt{nb}, is significantly smaller than the sample size, \texttt{bw.ucv} has the disturbing tendency to choose the very smallest bandwidth allowed. This is illustrated in Listing \ref{output1}, where we show the output of the \texttt{bw.ucv} function applied to a sample of size $n=10^6$ drawn from a standard normal and the number of bins set to its default value of $\texttt{nb}=1,000$. In this case the true cross-validation bandwidth is approximately 0.06, while \texttt{bw.ucv} returned a much smaller smoothing parameter, the lower bound of the search interval.

\begin{lstlisting}[frame=tb,caption=Bad behaviour of \texttt{bw.ucv} when using the default number of bins,captionpos=b,label=output1]
set.seed(1)
x = rnorm(10^6)
bw.ucv(x,lower=0.001,upper=1)

[1] 0.001045393
Warning message:
In bw.ucv(x, lower = 0.001, upper = 1) :
  minimum occured at one end of the range
\end{lstlisting}

For some densities, \texttt{bw.ucv} works fine with \texttt{nb} being relatively small with respect to the sample size. However, for more complex, heavy-tailed or multimodal, densities, \texttt{nb} needs to be quite close to the sample size for \texttt{bw.ucv} to give sensible results. This limits the computational gain that binned cross-validation could in principle achieve. 
Even when $\texttt{nb}$ is equal to the sample size, \texttt{bw.ucv} returns an incorrect value in a small proportion of cases.
In spite of this, in practice, we recommend using \texttt{bw.ucv} with $\texttt{nb}$ close to the sample size. Taking this suggestion into account, if $\texttt{nb} = m$ at the subsample level for $\hat{h}(m,N)$, we found that the average of the bagged bandwidths obtained using \texttt{bw.ucv} is usually quite close to the results obtained employing the more accurate non-binned version of $\hat{h}(m,N)$. Moreover, by using \texttt{bw.ucv} in the implementation of $\hat{h}(m,N)$, its runtime can obviously be significantly reduced, even being much shorter than the time needed for the computation of the binned cross-validation selector, especially for large sample sizes and certain values of $m$ and $N$. This can be observed in Table \ref{tab:times}, which shows the computing time for the binned version of leave-one-out 
cross-validation and the bagged bandwidth selector for different values of $n$, $m$ and $N$. For $\hat{h}(m,N)$, we considered $\texttt{nb} = m$ at the subsample level and the code was run in parallel on an Intel Core i5-8600K 3.6GHz using the R package \texttt{baggedcv} \citep{baggedcv}. In the case of the binned version of leave-one-out cross-validation, the number of bins was also set equal to $n$ to provide a fair comparison of both methods. As we
can see the bagged
bandwidth can achieve a significant reduction in computing time with
respect to binned leave-one-out cross-validation for samples of considerable size.

\begin{table}[h]
	\begin{center}
\begin{threeparttable}
	\def~{\hphantom{0}}
	\caption{Elapsed time, in seconds, for binned leave-one-out cross-validation and the bagged bandwidth selector.	\label{tab:times}}
		\begin{tabular}{ccccc}
			&&&Bagged CV&\\
			\cline{3-5}
			&&$m=1,000$&$m=5,000$&$m=10,000$\\
			$n$ & bw.ucv( . , nb=n) &$N=500$&$N=500$&$N=500$ \\ 
			\hline
			$10^5$ & 3.1 & 1.1 & 2.0 & 4.3\\
			$10^6$ & 367 & 1.3 & 2.2 & 4.4\\
			\hline
	\end{tabular}
	\begin{tablenotes}
	\item Computing time for
		bagged cross-validation depends on $m$, $N$ and the number of CPU cores.
	\end{tablenotes}
	 \end{threeparttable}
 \end{center}
\end{table}


In addition to the substantial reduction in computing time, the bagged cross-validation bandwidth 
yielded, in general, greater statistical precision. This can be observed in Figure \ref{img:hhat_supp}, where the sampling distributions of $\log \hat{h}_n/h_{n0}$ and $\log \hat{h}(m,N)/ h_{n0}$, for different values of $m$, for models D1, in the left panel, and D2, in the  right panel, are presented.
Specifically, we considered, for D1, the values of $m$: $5,000$,
  $13,081$ ($m_0$), $20,000$, and $\hat m_0$ computed with $s=50$ and
  $r=500, 1,000, 5,000$, while, for D2, the values of $m$ employed
  were:  $5,000$, $20,326$ ($m_0$), $25,000$, and $\hat m_0$ computed
  with $s=50$ and $r=500, 1,000, 5,000$. 
 It is clear that the bagged bandwidth achieves, in general, an important reduction in the mean squared error with respect to the leave-one-out cross-validation selector. Namely, the bagged bandwidth with $m = m_0$ produced a mean squared error which is $95.3\%$ and $92.2\%$ lower than that of the leave-one-out cross-validation bandwidth for models D1 and D2, respectively.
This significant reduction is also observed, in general, when using $m=\hat m_0$ for each simulated sample. In that case, for $r = 1,000$ ($r = 5,000$), the mean squared error reduction with respect to leave-one-out cross-validation is $95.9\% $ ($95.9 \% $) for model D1 and $92.3\% $ ($93.5\% $) for model D2. 

\begin{figure}[h]
	\includegraphics[width=0.5\textwidth]{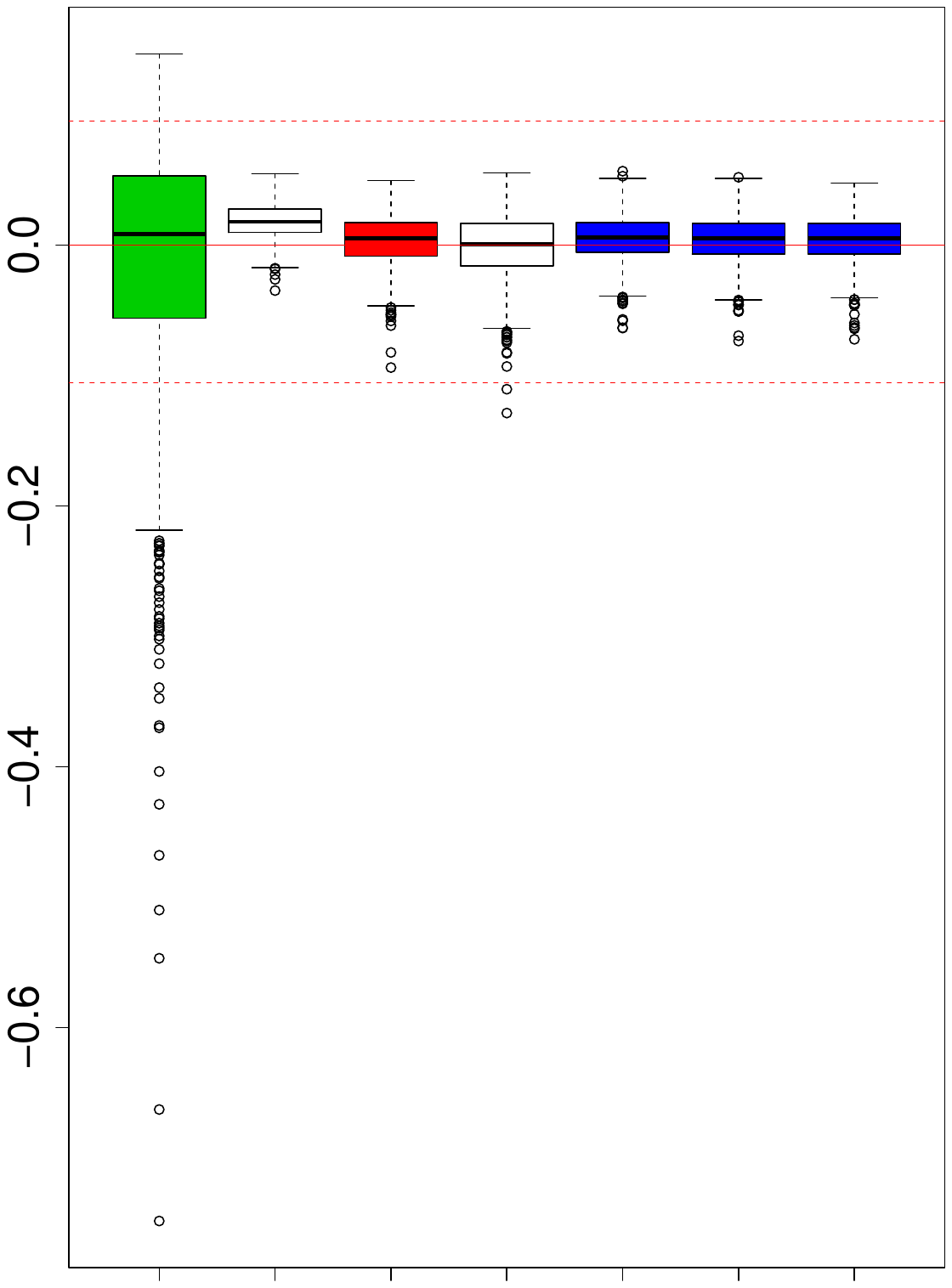}
	\includegraphics[width=0.5\textwidth]{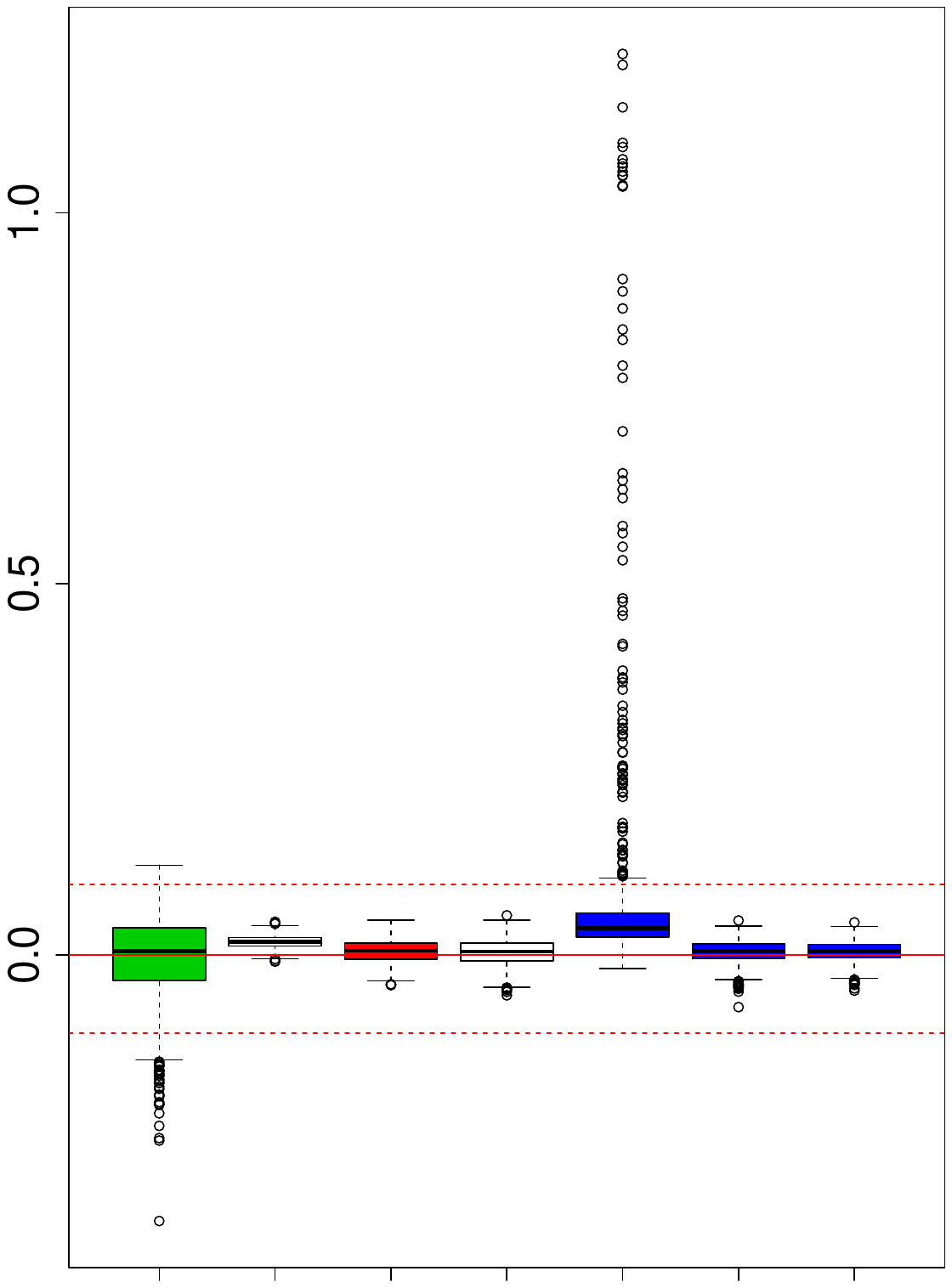}
	\caption{Sampling distribution of
		$\log \hat{h}/h_{n 0}$, with $\hat{h}$
		denoting the leave-one-out cross-validation (green) and the
		bagged bandwidths for different values of $m$. For the
		bagged bandwidths, we considered $N = 500$ and $m \in
		\{5,000, 13,081$ (red)$, 20,000, \hat m_0$ (blue)$\}$, for
		density D1 (left panel); and $m \in \{5,000, 20,326$ (red)$,
		25,000, \hat m_0$ (blue)$\}$, for density D2 (right
		panel). 
		The two white boxes correspond, from left to right, to $m=5,000$ and $20,000$, for D1 (left panel); and to $m=5,000$ and $25,000$, for D2 (right panel).
		The three blue boxes correspond, from left to
		right, to  $r = 500, 1,000, 5,000$.
		Red dotted lines are plotted
at values 0.9 and 1.1 for reference.
		\label{img:hhat_supp}} 
\end{figure}

The mean squared error of the bagged bandwidth using $m=\hat m_0$ may be larger than the one for leave-one-out cross-validation for density D2 using $r = 500$, as it can be observed in the left blue box-plot on the right panel in
Figure \ref{img:hhat_supp}. These results are
somewhat misleading because the final behaviour of the kernel density estimator with the bagged bandwidth selector, denoted by $\hat h$ for simplicity, is still very good in this setting. The distribution of $\hat h$ is biased upward
and there are numerous extremely large values of $\hat h$. However,
it turns out that even the largest of these bandwidths produces very
effective density estimates, as observed in Figure \ref{img:clawest}. 
Consider, for example, $\log \hat
h/h_{n0}=1$, which means that $\hat h \simeq 2.72h_{n0}$. In Figure
\ref{img:clawest}, we provide the claw 
density and two kernel estimates from a sample of size
$10^5$. The bandwidths of the two estimates are $h_{n0}=0.031$ and
$2.72h_{n0} \simeq 0.084$. The kernel estimate with larger bandwidth captures
the five modes and has better tail behaviour than the estimate based
on the MISE bandwidth. Figure \ref{img:clawest} illustrates the fact
that integrated squared error (ISE) loss is not always ideal. One might well
prefer an estimate with larger than optimum ISE, as long as it captures
all the important features of the underlying density and is smoother
than the ISE optimal estimate. However, despite this remark, ISE error criterion can be used to see the effect of the different bandwidth selectors on the kernel density estimates.
Figure \ref{img:ises} shows the sampling distribution of ratio of the ISE of the kernel density estimates using the bagging cross-validation
bandwidths and the classical cross-validation one, ${\rm ISE}\{\hat h(m,N)\}/{\rm ISE}(\hat{h}_n)$, for both models and the same values of $m$ considered in Figure \ref{img:hhat_supp}. In this case, outliers were omitted in order to be able to appreciate the differences between the different box-plots.

\begin{figure}[h]
	\includegraphics[width=1.0\textwidth]{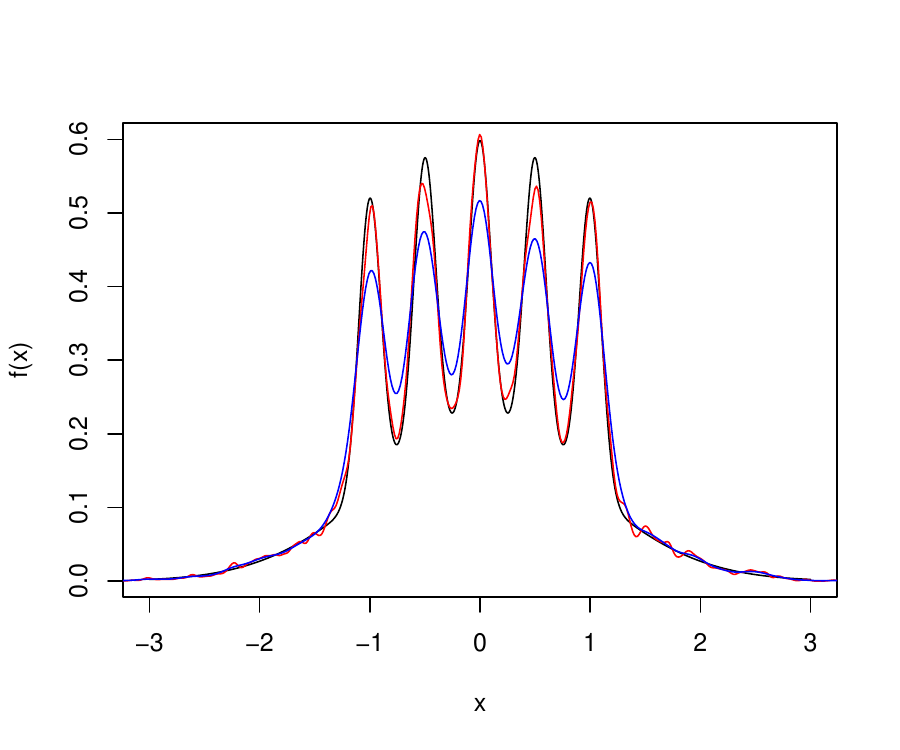}
	\caption{Claw density (black line) and kernel estimates 
		(red and blue lines). The kernel
		estimates are computed from a sample of size $10^5$. The
		red estimate uses the MISE optimal bandwidth of 0.031 and the
		blue uses bandwidth $0.084$.\label{img:clawest}} 
\end{figure}

\begin{figure}[h]
	\includegraphics[width=0.5\textwidth]{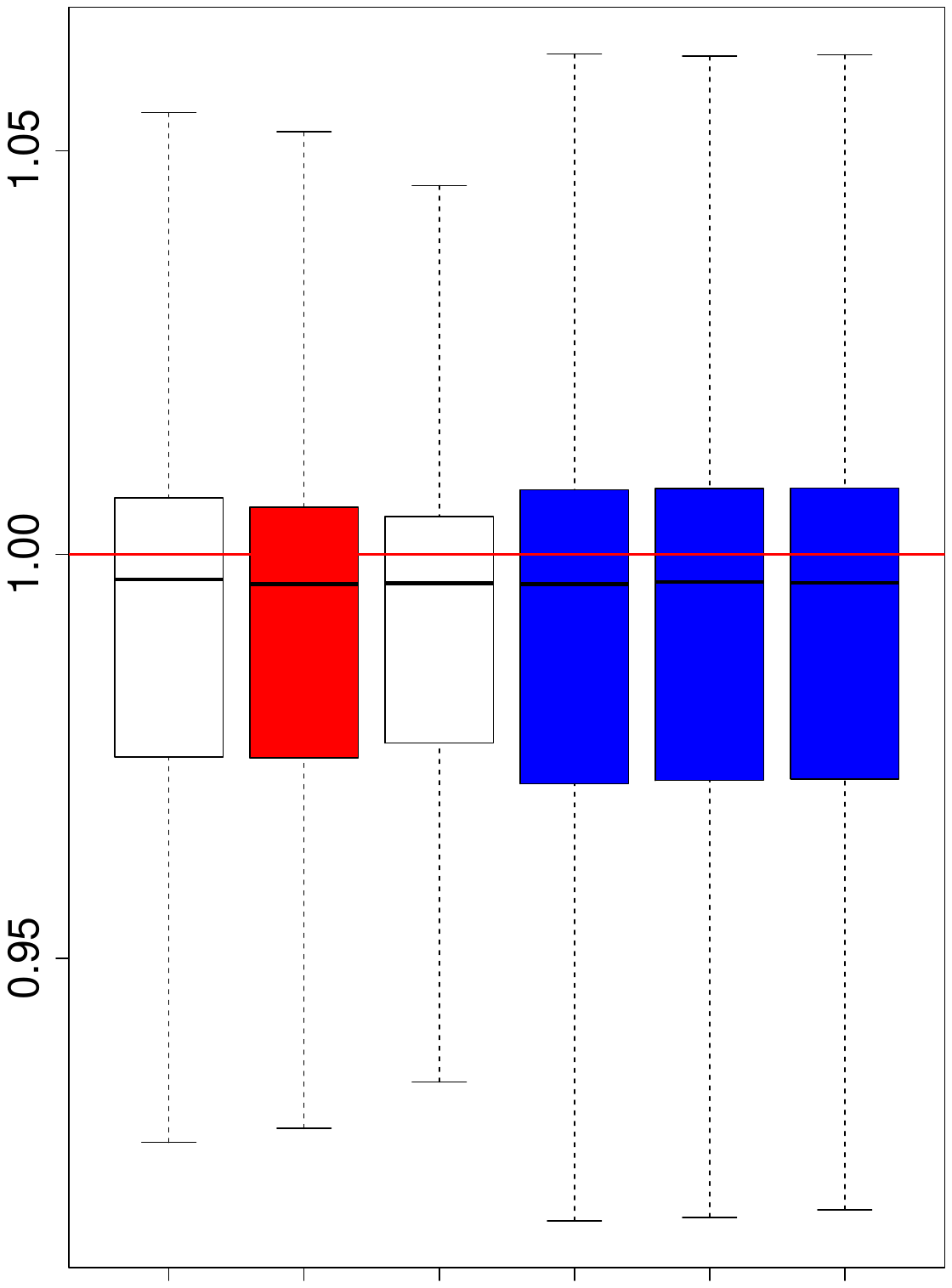}
	\includegraphics[width=0.5\textwidth]{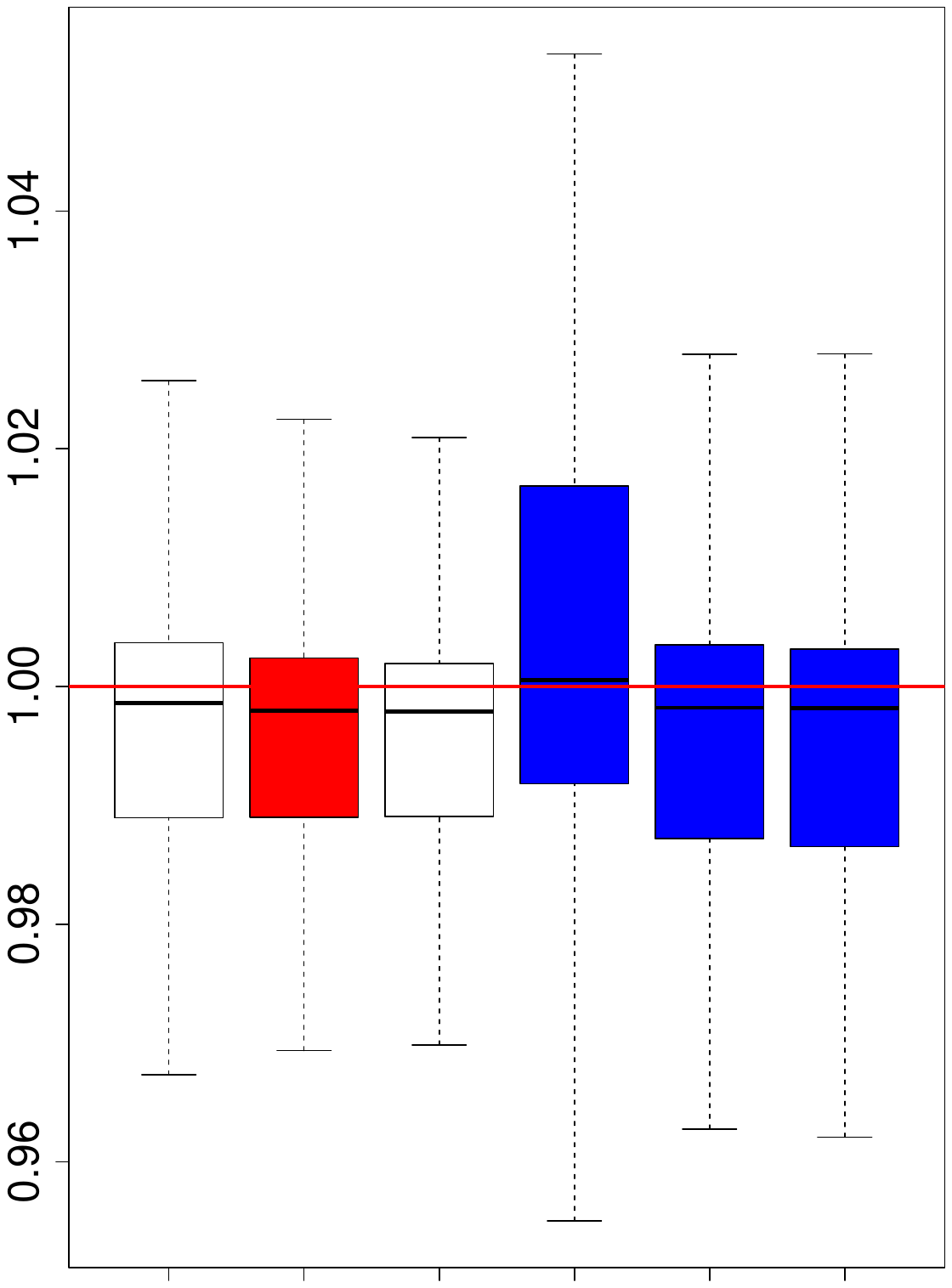}
	\caption{Sampling distribution of the random variable ${\rm ISE}\{\hat h(m,N)\}/{\rm ISE}(\hat{h}_n)$, with $\hat h(m,N)$ denoting the
		bagged bandwidths for different values of $m$, and $\hat{h}_n$ denoting the leave-one-out cross-validation bandwidth. For the
		bagged bandwidths, we considered $N = 500$ and $m \in
		\{5,000, 13,081$ (red)$, 20,000, \hat m_0$ (blue)$\}$, for
		density D1 (left panel); and $m \in \{5,000, 20,326$ (red)$,
		25,000, \hat m_0$ (blue)$\}$, for density D2 (right
		panel). 
		The two white boxes correspond, from left to right, to $m=5,000$ and $20,000$, for D1 (left panel); and to $m=5,000$ and $25,000$, for D2 (right panel).
		The three blue boxes correspond, from left to
		right, to  $r = 500, 1,000, 5,000$.}
		\label{img:ises} 
\end{figure}

The means of ${\rm ISE}\{\hat h(m,N)\}/{\rm ISE}(\hat{h}_n)$ for the values of $m$ and $N$, and densities considered in Figure \ref{img:ises}, as well as the proportion of times where the ISE of the kernel density estimates using $\hat h(m,N)$ is lower than using $\hat{h}_n$ are shown in Table \ref{tab:means_ise}.
In general, it can be observed a slightly better performance of the estimators when using the bagged bandwidths than when employing the leave-one-out cross-validation selector, except when considering the density D2 and using $m=\hat m_0$, with $r = 500$ (left blue box-plot on the right panel in
Figure \ref{img:ises}). These results are totally consistent with those shown in Figure \ref{img:hhat_supp} for the bandwidths.

\begin{table}[h]
	\begin{center}
\begin{threeparttable}
	\def~{\hphantom{0}}
	\caption{Top Table: means of ${\rm ISE}\{\hat h(m,N)\}/{\rm ISE}(\hat{h}_n)$, with $\hat h(m,N)$ computed using the combinations of $m$ and $N$ and densities considered in Figure \ref{img:ises}. Bottom Table: proportion of values of $\hat h(m,N)$ whose {\rm ISE} is lower than that of $\hat{h}_n$.}
		\begin{tabular}{ccccccc}
			\\
			&&&Means&&&\\
			\hline
			Density & $B_1$ & $B_2$ & $B_3$ & $B_4$ & $B_5$ & $B_6$\\
			D1 & $0.98533$ & $0.98505$ & $0.98512$ & $0.98448$ & $0.98428$ & $0.98537$\\
			D2 & $0.99624$ & $0.99594$ & $0.99530$ & $1.23742$ & $0.99539$ & $0.99494$\\
			\hline\\
			&&&Proportions&&&\\
			\hline
			Density & $B_1$ & $B_2$ & $B_3$ & $B_4$ & $B_5$ & $B_6$\\
			D1 & 0.606 & 0.603 & 0.609 & 0.590 & 0.593 & 0.590\\
			D2 & 0.584 & 0.622 & 0.637 & 0.461 & 0.604 & 0.599\\
			\hline
	\end{tabular}
	\begin{tablenotes}
	\item $B_i$ refers to the $i$th box-plot in order of appearance in Figure \ref{img:ises}.
		\end{tablenotes}
	 \end{threeparttable}
	\label{tab:means_ise}
	\end{center}
		\end{table}

In Figure \ref{img:mhat_supp}, the sampling distribution of $\hat m_0/m_0$ is shown. 
It can be observed that the mean squared error
of $\hat m_0$ is reduced as $r$ increases. Furthermore, the bias of
the estimator depends on the complexity of the target density. 
For small values of $r$, in spite of the high variability of $\hat m_0$, the sampling distribution of the bagged bandwidth, considering $m = \hat m_0$, is virtually unchanged with respect to the case $m = m_0$ for densities that are not very complex, such as D1. For more complex densities, such as D2, the effect that the variability of $\hat m_0$ has on the bagged bandwidth is more noticeable for small values of $r$, translating into a more biased bandwidth.
More
importantly, when we compare the errors in Figure
\ref{img:mhat_supp} and Figure \ref{img:hhat_supp}, it is clear that
there is a large range of values for $m$ around its optimal value,
$m_0$, such that the effect the error of $\hat m_0$ has on the
sampling distribution of $\hat h(\hat m_0, N)$ is very small.

\begin{figure}[h]
	\includegraphics[width=0.5\textwidth]{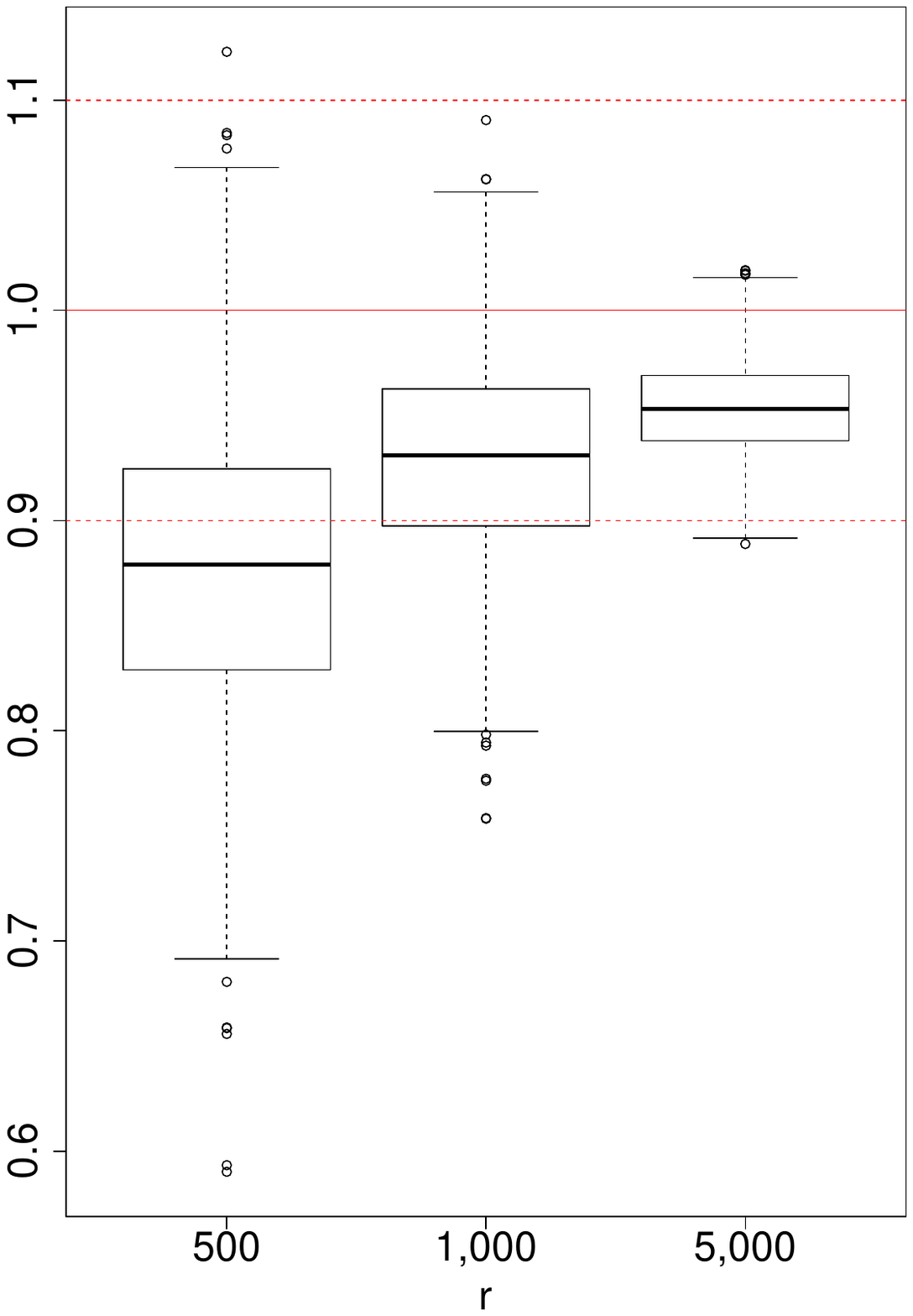}
	\includegraphics[width=0.5\textwidth]{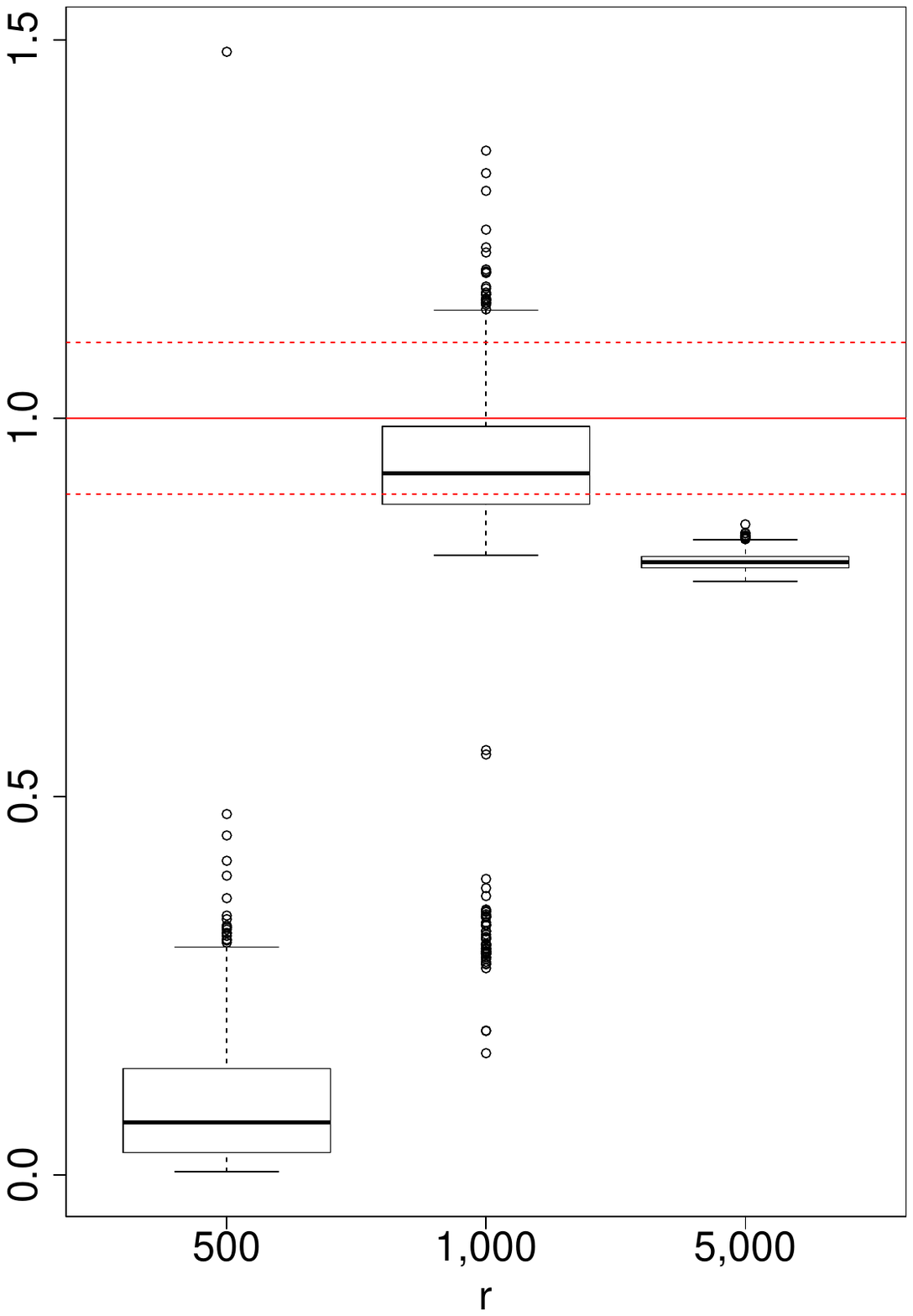}
	\caption{Sampling distribution of $\hat m_0/m_0$, with
		$\hat m_0$ denoting the estimator of the optimal subsample
		size, $m_0$, as defined in Section $4$ of the main paper,
		for densities D1 (left panel) and D2 (right panel). The
		values chosen for the parameters of the estimator were $s
		= 50$ and, from left to right, $r \in \{500, 1,000,
		5,000\}$.
		Red dotted lines are plotted
at values 0.9 and 1.1 for reference.
		\label{img:mhat_supp}} 
\end{figure}

\section*{Appendix 3. Real data example}
\label{sec:apli}
To further explore the performance of the bagged bandwidth selector, we considered the public dataset ``On-Time: Reporting Carrier On-Time Performance'' corresponding to the year $2017$, available at 
\url{https://www.transtats.bts.gov/Fields.asp}. 
In particular, we were interested in the variable \texttt{ArrDelay}, which measures the difference in minutes between scheduled and actual arrival time. Early arrivals show negative numbers. Due to the fact that the dataset contains many ties and in order to avoid problems when performing cross-validation, we decided to remove the ties by jittering the data. In particular, we worked with the sample of size $n = 5,579,346$ which results from adding a random sample of size $n$, drawn from a continuous uniform distribution defined on the interval $(-0.5, 0.5)$, to the original dataset.

To estimate the optimal subsample size, $m_0$, for the bagged bandwidth, we used the procedure described in Section $4$ of the main paper, considering $N = 100$ subsamples. In particular, using $r = 1,000$ and $s = 500$, yielded the estimate $\hat m_0 = 272,222$. The process of estimating $m_0$ with those parameters took $32$ seconds. The estimated bagged bandwidth with these values of $m$ and $N$ was $\hat h(m=272,222, N=100) = 0.490$. Its calculation took $63$ seconds. The calculation of both $\hat m_0$ and $\hat h(m, N)$ were executed in parallel on an Intel Core i5-8600K 3.6GHz. Figure \ref{img:density_hbag_bwucv} shows the kernel density estimates obtained when considering the bagged bandwidth $h = \hat h(\hat m_0, N=100)=0.490 $ and the bandwidth produced by the R function \texttt{bw.ucv}, using the same number of bins and search interval as in the case of the bagged bandwidth, that is, $h =$ \texttt{bw.ucv($\cdot$, nb=1e5, lower=0.01, upper=1)}, that returned the value 0.01039. As we can see, even with those parameters, \texttt{bw.ucv} basically returns the lower bound of the search interval thus producing a heavily undersmoothed estimate of the underlying density.

\begin{figure}[h]
	\includegraphics[width=0.5\textwidth]{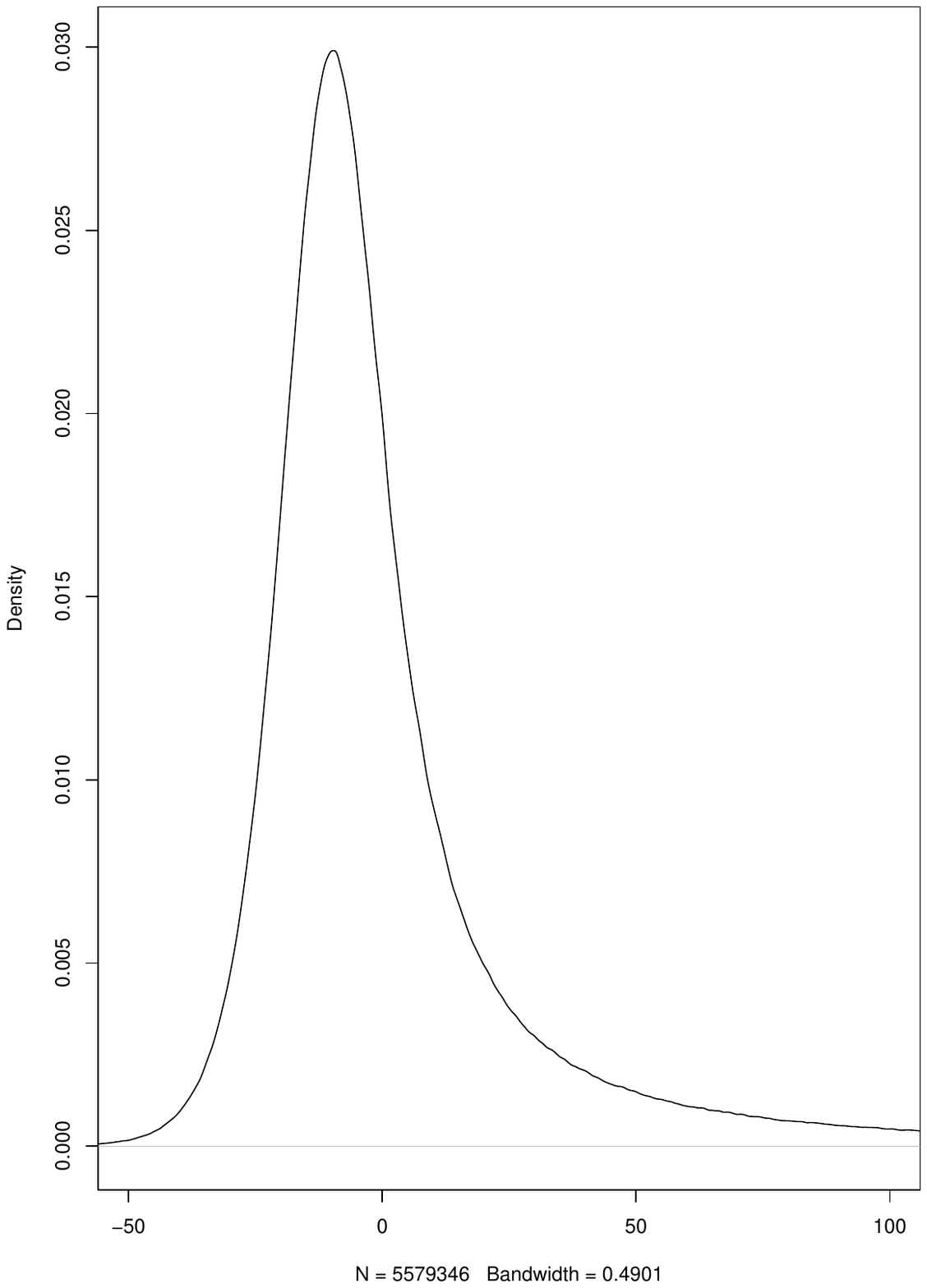}
	\includegraphics[width=0.5\textwidth]{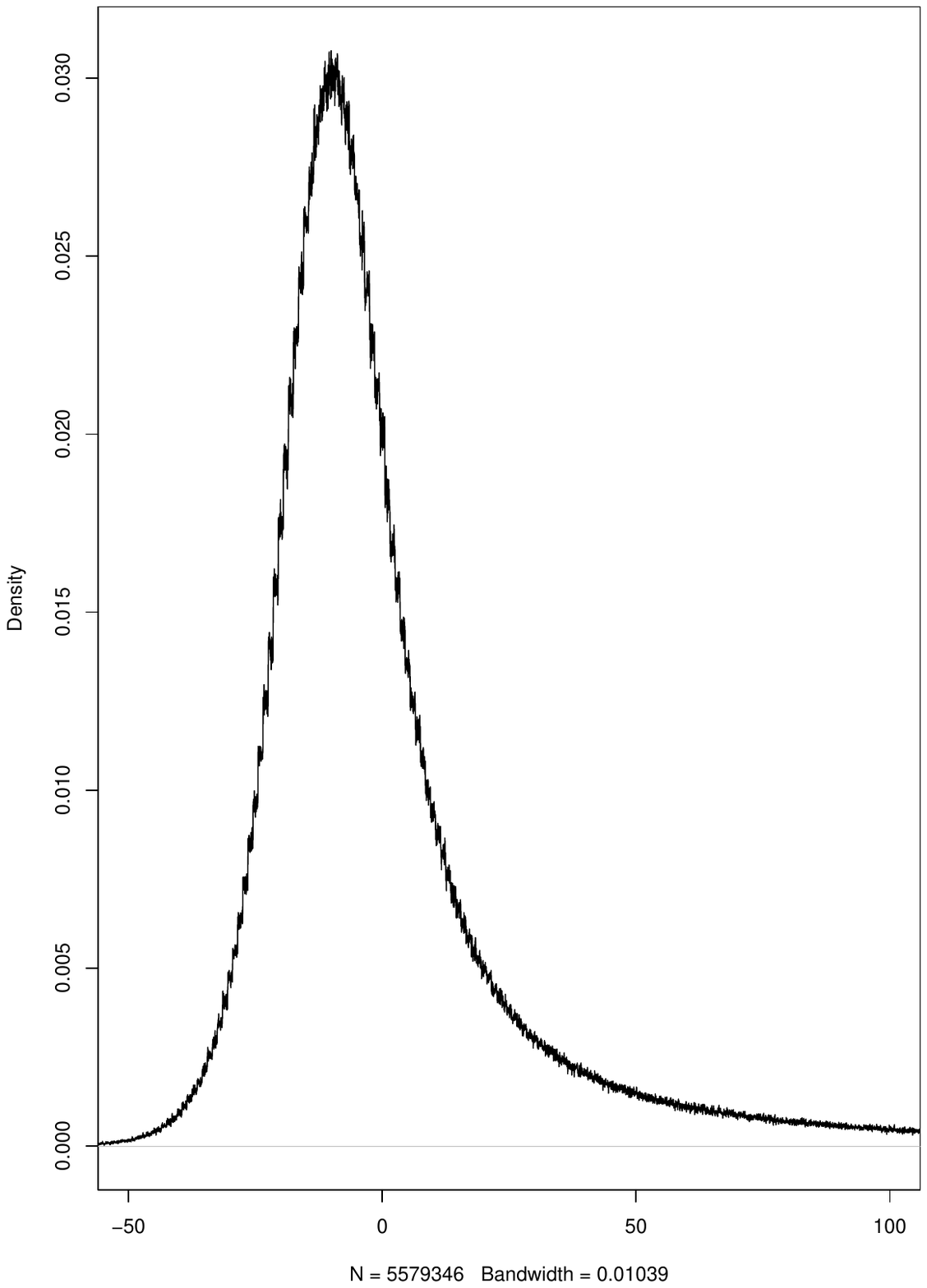}
	\caption{Kernel density estimates with bandwidths $h = \hat h(\hat m_0, N=100)$ (left) and $h =$ \texttt{bw.ucv($\cdot$, nb=1e5, lower=0.01, upper=1)} (right).}
	\label{img:density_hbag_bwucv} 
\end{figure}

Computing the leave-one-out cross-validation bandwidth for the
  whole sample is prohibitive due to the huge amount of time it would
  require. Even with a binned implementation, as employed in the R
  function \texttt{bw.ucv}, the computing time would be very high, as
  highlighted in Section $2$. In order for this function to produce accurate
  results the number of bins must be very close to $n$. Therefore, to
  predict the value of the cross-validation bandwidth for the original
  sample size, $n$, and also the time required for its computation, we
  used appropriate regression models. We repeated these experiments
  considering binned and non-binned cross-validation bandwidths. The
  predicted cross-validation bandwidth for the whole
  sample is practically identical whether or not one uses binning, with a
  large enough number of bins, and hence we just describe the
  experiment when using a binned implementation. Nevertheless, the
  predicted time is obviously much higher when binning is not used, as
  we will see later. Specifically, we selected 100 subsamples of
  sizes $557$, $5,579$ and $55,793$ from the whole dataset. For each
  size and subsample, we computed the binned version of the
  leave-one-out cross-validation bandwidth, using the R function
  \texttt{bw.ucv} with \texttt{nb}, number of bins, equal to the
  corresponding sample size (see Figure
  \ref{img:box_hcv_real}). Finally, we considered the parametric
  regression model: 
\beq\label{reg_mod}
Y_i = \beta_0 n_i^{\beta_1},
\eeq
where $n_i \in \{557,\, 5,579,\, 55,793\}$ and $Y_i \in \{3.606, 2.129, 1.352\}$ denotes the mean of the binned cross-validation bandwidths using the subsamples of size $n_i$. Taking logarithms in \eqref{reg_mod}, we get a linearized version of \eqref{reg_mod},
\beq\label{reg_mod_log}
\log Y_i = \log \beta_0 + \beta_1 \log n_i,
\eeq
which we can see as a linear regression model with parameters $\log \beta_0$, intercept, and $\beta_1$, slope. Applying least squares, we obtained the following estimates for the parameters of model \eqref{reg_mod}:
\beqn
\hat \beta_0 &=& 13.69,\\
\hat \beta_1 &=& -0.213.
\eeqn

With these values of $\hat \beta_0$ and $\hat \beta_1$, the predicted value of the leave-one-out cross-validation bandwidth for the original sample size is $\hat h_{n} = 0.501$, very close to the value produced by the bagged approach, $\hat h(m=272,222, N=100) = 0.490$. Figure \ref{img:hcv_reg_fit} shows the fitted values for the nonlinear model defined in \eqref{reg_mod}. Analogously, we considered a model similar to the one described in \eqref{reg_mod} to predict the time required to compute a binned version of the ordinary cross-validation bandwidth for the original sample. Fitted values for this model are shown in Figure \ref{img:fitted_values_tmod}. As previosuly, we employed the R function \texttt{bw.ucv} with \texttt{nb} equal to the corresponding sample size to compute the different cross-validation bandwidths. In this case and using the same notation as in \eqref{reg_mod}, we considered $n_i = \{5,579,\, 55,793,\, 557,934\}$ and $Y_i = \{0.0102, 0.959, 103.08\}$, with $Y_i$ now denoting the elapsed time (in seconds) needed to compute \texttt{bw.ucv($\cdot$, nb=$n_i$)}, that is, the binned cross-validation bandwidth for a sample of size $n_i$ with the number of bins set to $n_i$. Again, using the same notation as in \eqref{reg_mod}, we obtained the following estimates for the model parameters:
\beqn
\hat \beta_0 &=& 3.14 \times 10^{-10},\\
\hat \beta_1 &=& 2.002.
\eeqn

This means that the time needed to compute the binned cross-validation bandwidth for the original sample is predicted to be approximately $2.8$ hours. Analogously, we repeated the experiment to predict the time required to compute a non-binned leave-one-out cross-validation bandwidth for the whole sample and this predicted time turned out to be $5.1$ years. Fitted values for the model are shown in Figure \ref{img:fitted_values_tmod_nonbinned}.

\begin{figure}[h]
	\includegraphics{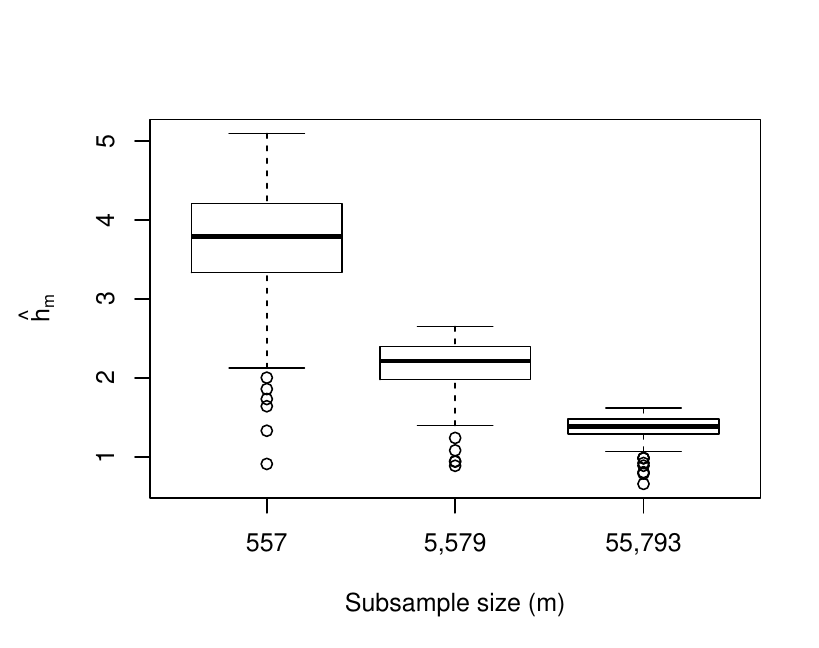}
	\caption{Box-plots of $\hat h_{m}$ for subsamples of size $m \in \{557,\, 5,579,\, 55,793\}$.}\label{img:box_hcv_real} 
\end{figure}

\begin{figure}[h]
	\includegraphics{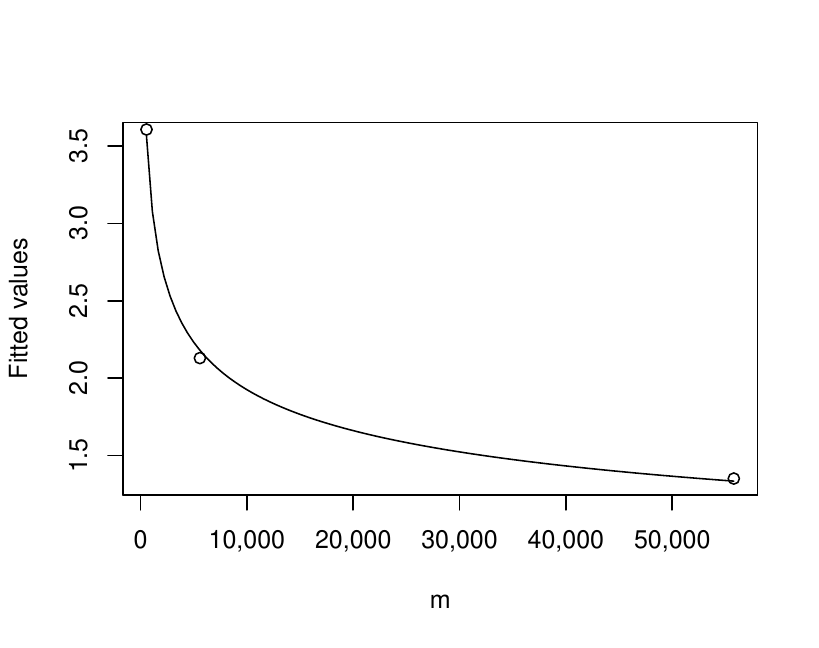}
	\caption{Fitted values for the regression model defined in \eqref{reg_mod}. White dots correspond to the observations used to fit the model.}
	 \label{img:hcv_reg_fit} 
\end{figure}

\begin{figure}[h]
	\includegraphics{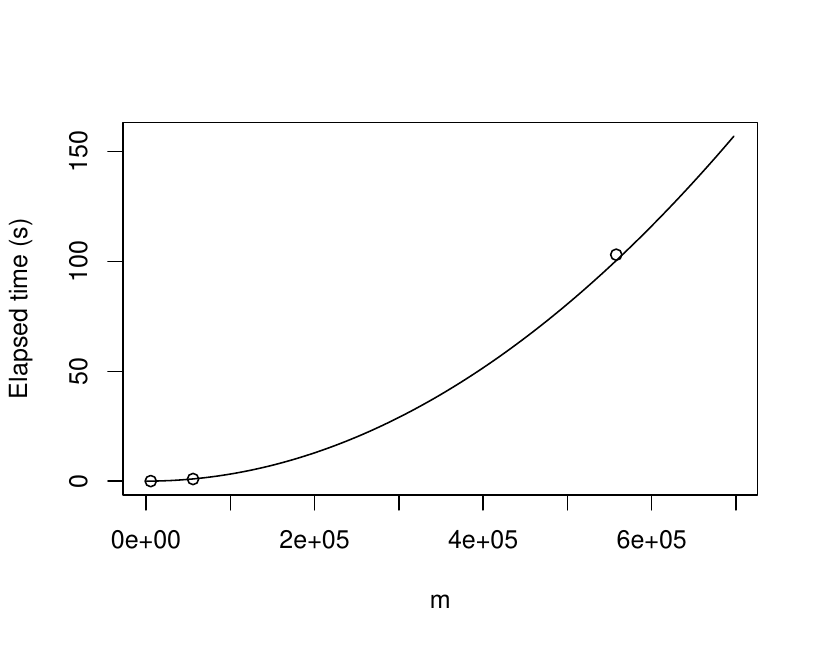}
	\caption{Fitted values for the regression model that relates the elapsed time needed to compute the binned cross-validation bandwidth to the sample size. White dots correspond to the observations used to fit the model.}
	\label{img:fitted_values_tmod}
\end{figure}

\begin{figure}[h]
	\includegraphics{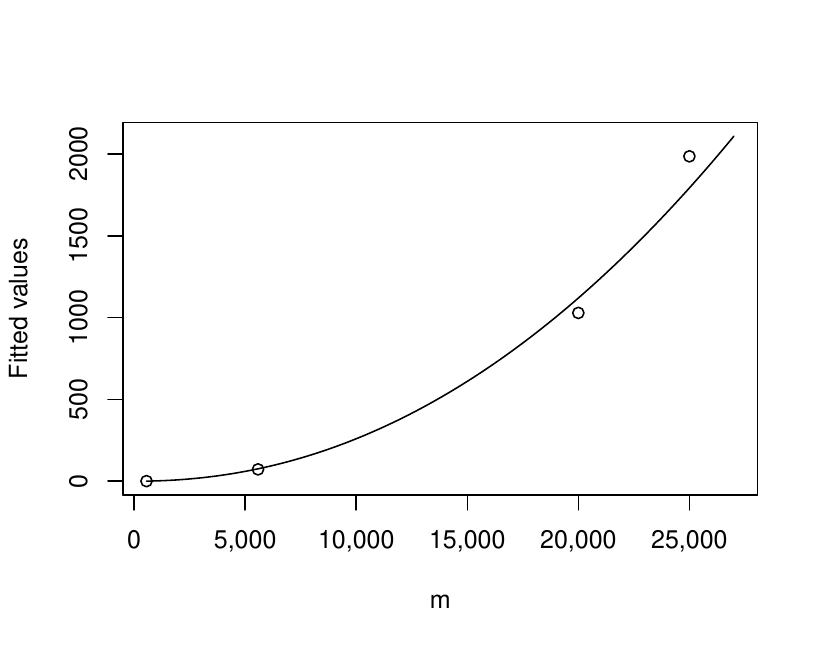}
	\caption{Fitted values for the regression model that relates the elapsed time needed to compute the standard non-binned cross-validation bandwidth to the sample size. White dots correspond to the observations used to fit the model.}
	\label{img:fitted_values_tmod_nonbinned} 
\end{figure}

\section*{Appendix 4. Variance of the bagged bandwidth (Hall and Robinson, 2009)}
In this section, we show that the variance approximation of the
  bagged bandwidth studied in \cite{HR} is in error. We also
  provide  the correct expression for this variance and the
  corresponding proof. 
The bagged bandwidth studied in \cite{HR}, $\hat h_{bagg}$,
corresponds to the case where $N = \infty$ in the smoothing
  parameter (3) of the main paper, that is, $\hat h_{bagg} = \hat
h(m,\infty)$, following the notation adopted. From equation (7) in the main paper, it follows that 
\beq\label{eq:var} 
\var\left(\hat h_{bagg}\right) = AC^2 m^{9/5}n^{-12/5}+o\left(m^{9/5}n^{-12/5}\right),
\eeq
which exactly matches the second term given in equation (13) of
\cite{HR}. However, it is claimed in that paper that the
dominant term is of order $m^{4/5}n^{-7/5}$. We will prove that this
last statement is wrong and that, in fact, the dominant term is
precisely the one given in \eqref{eq:var}. 

It can be easily proved that for a sample of size $n$ we have
\beq\label{eq:CV}
CV(h) = M(h)-R(f)+S(h),
\eeq
where, as previously, $M(h)$ denotes the MISE function of
  the Parzen--Rosenblatt kernel density estimator for a sample of
size $n$, and $S(h) = S_1(h)+S_2(h)$ is defined on p.~$184$ of
\cite{HR}. From \eqref{eq:CV} it follows that, for any $r \in
\mathbb{N}$, 
\beqn
\var\left\{CV^{(r)}(h)\right\} = \var\left\{S^{(r)}(h)\right\}.
\eeqn

More importantly, finding the asymptotic variance of the cross-validation bandwidth, whether bagged or ordinary, boils down to finding $\var\left\{S'(h)\right\}$. As stated in equation (\ref{cvr}) in Section 
\ref{sec:teo} 
of this document,  for any $r \geq 1$,
\beqn
CV^{(r)}(h) = M^{(r)}(h)+\frac{1}{n(n-1)}\sum\limits_{i \neq j}\bar \gamma_{nh}^{(r)}(X_i-X_j),
\eeqn
where 
$$\bar \gamma_{nh}^{(r)}(u) = \gamma_{nh}^{(r)}(u)-E\left\{\gamma_{nh}^{(r)}(X_1-X_2)\right\},$$
 $$\gamma_{nh}^{(r)}(u) = \frac{{\rm d}^r\gamma_{nh}(u)}{{\rm d}h^r},$$
 $$\gamma_{nh}(u) = \gamma_n(u/h)/h$$ 
 and 
 $$\gamma_n(u) = \frac{n-1}{n}K*K(u)-2K(u).$$ 
Therefore,
\beq\label{eq:varSprime}
\var\left\{S'(h)\right\} = \frac{1}{n^4h^2}\var\left\{\sum\limits_{i \neq j}H(X_i-X_j)\right\},
\eeq
where 
$$H(u) = \gamma_{e,h}(u)+u(\gamma_{e,h})'(u),$$
$$\gamma_{e,h}(u) = \gamma_e(u/h)/h$$ 
and 
$$\gamma_e(u) = \frac{n}{n-1}\gamma_n(u).$$
Let us now define $\tilde H(u) = \gamma_e(u)+u\gamma_e'(u)$, so we have that $H(u) = \tilde H_h(u)$. Standard algebra gives
\beq\label{eq:varsum}
\var\left\{\sum\limits_{i \neq j}H(X_i-X_j)\right\} = 4n(n-1)(n-2)C_b+2n(n-1)C_c,
\eeq
where $C_b = \cov\left\{H(X_1-X_2), H(X_1-X_3)\right\}$ and $C_c = \var\left\{H(X_1-X_2)\right\}$. These terms can be further decomposed into 
\beq\label{eq:Cb}
C_b = C_{b1}-C_{b2}^2
\eeq
and
\beq\label{eq:Cc}
C_c = C_{c1}-C_{b2}^2,
\eeq
where
\beqn
C_{b1} &=& \int H*f(x)^2f(x)\, {\rm d}x,\\
C_{c1} &=& \int H^2*f(x)f(x)\, {\rm d}x,\\
C_{b2} &=& \int H*f(x)f(x)\, {\rm d}x.
\eeqn

Using the facts that $\tilde H$ is symmetric, $\mu_0\left(\tilde H\right) = 0$, $\mu_2\left(\tilde H\right) = 4\mu_2(K)/(n-1)$, and $\mu_4\left(\tilde H\right) = \mu_6\left(\tilde H\right) = O\left(1\right)$, we have
\beqn
C_{b2} &=& \int\int \frac{1}{h}\tilde H\left(\frac{x-y}{h}\right)f(y)f(x)\, {\rm d}x\,{\rm d}y = \int\int \tilde H(u)f(x-hu)f(x)\, {\rm d}x\, \rm{d}u\\
&=& \int\int \tilde H(u)\left\{f(x)-huf'(x)+\cdots-\frac{h^5u^5}{5!}f^{(5)}(x)+\frac{h^6u^6}{6!}f^{(6)}(\tilde x)\right\}f(x)\, {\rm d}x\,{\rm d}u\\
&=& \int f(x)\left\{\frac{h^2}{2}\mu_2\left(\tilde H\right)f''(x)+\frac{h^4}{4!}\mu_4\left(\tilde H\right)f^{(4)}(x)+O\left(h^6\right)\right\}\, {\rm d}x\\
&=& \frac{1}{4}\mu_2(K)^2R(f'')h^4+O\left(h^6\right),
\eeqn
and, therefore,
\beq\label{eq:Cb22}
C_{b2}^2 = \frac{1}{16}\mu_2(K)^4R(f'')^2 h^8+O\left(h^{10}\right).
\eeq

For the term $C_{b1}$,
\beq\label{eq:Cb1}
C_{b1} &=& \int f(x)\left\{\int \frac{1}{h}\tilde H\left(\frac{x-y}{h}\right)f(y)\, dy\right\}^2\, {\rm d}x = \int f(x)\left\{\int \tilde H(u)f(x-hu)\, du\right\}^2\, {\rm d}x\nonumber \\
&=& \int f(x)\left\{\frac{1}{4}\mu_2(K)^2f^{(4)}(x)h^4+O\left(h^6\right)\right\}^2\, {\rm d}x \nonumber \\ 
&=& \int f(x)\left\{\frac{1}{16}\mu_2(K)^4f^{(4)}(x)^2h^8+O\left(h^{10}\right)\right\}\, {\rm d}x \nonumber \\
&=& \frac{1}{16}\mu_2(K)^4 J_1 h^8+O\left(h^{10}\right),
\eeq
where $$J_1 = \int f^{(4)}(x)^2f(x)\, {\rm d}x.$$

The term $C_{b1}$ can be handled in a similar way
\beq\label{eq:Cc1}
C_{c1} &=& \frac{1}{h^2}\int\int \tilde H\left(\frac{x-y}{h}\right)^2 f(y)f(x)\, {\rm d}x\,{\rm d}y = \frac{1}{h}\int\int \tilde H(u)^2f(x-hu)f(x)\,{\rm d}x\,{\rm d}u\nonumber\\
&=& \frac{1}{h}\int\int \tilde H(u)^2f(x)\left\{f(x)-huf'(x)+\frac{h^2u^2}{2}f''(\tilde x)\right\}\,{\rm d}x\,{\rm d}u\nonumber\\
&=& \frac{R(f)R\left(\tilde H\right)}{h}+O\left(h\right).
\eeq

Plugging \eqref{eq:Cb22}, \eqref{eq:Cb1} and \eqref{eq:Cc1} into \eqref{eq:Cb} and \eqref{eq:Cc} yields, respectively,
\beq\label{eq:Cbfinal}
C_b = \frac{1}{16}\mu_2(K)^4\left\{J_1-R(f'')^2\right\}h^8+O\left(h^{10}\right),
\eeq
\beq\label{eq:Ccfinal}
C_c = \frac{R(f)R\left(\tilde H\right)}{h}+O\left(h\right).
\eeq

Now, plugging \eqref{eq:Cbfinal} and \eqref{eq:Ccfinal} into \eqref{eq:varsum} and then into \eqref{eq:varSprime}, and using the fact that $2R\left(f\right)R\left(\tilde H\right) = A_3+O\left(n^{-1}\right)$, we get
\beq\label{eq:varSprimefinal}
\var\left\{S'(h)\right\} = A_3\frac{1}{n^2h^3}+O\left(\frac{1}{n^2h}\right),
\eeq
where $A_3$ is defined on p.~183 of \cite{HR}. Equation
\eqref{eq:varSprimefinal} is completely consistent with the results
obtained in \cite{HallMarron} and
\cite{scottterrell}. Now, taking variance in equation (A2)
  of \cite{HR} and plugging \eqref{eq:varSprimefinal} into
that expression yields \eqref{eq:var}. Equation
\eqref{eq:varSprimefinal} is enough to show that expression (A3) of
\cite{HR} is wrong, which in turn explains the error in
their equation (13) regarding the variance of the bagged
bandwidth. Nonetheless, we will provide an asymptotic expression for
$\var\left\{S_1'(h)\right\}$, since that is where the error in
\cite{HR} comes from. 

From the definition of $V_{nh}(X_i)$ and $S_1(h)$ given on p. 184 of
\cite{HR}, it is easy to show that 
\beqn
V_{nh}(X_1) = \left(1-n^{-1}\right)\tilde z_1^{(h)}-\tilde T_1^{(h)},
\eeqn
where
\beqn
\tilde z_1^{(h)} = K_h*K_h*f(X_1)-\int K_h*f(x)^2\,{\rm d}x
\eeqn
and
\beqn
\tilde T_1^{(h)} = 2\left\{K_h*f(X_1)-\int K_h*f(x)f(x)\,{\rm d}x\right\}.
\eeqn

Let us define the functions $\nu$ and $\eta$, where
\beqn
\nu(x) = K(x)+xK(x)
\eeqn
and
\beqn
\eta(x) = K*K(x)+x(K*K)'(x).
\eeqn
Then, we have that
\beqn
\frac{{\rm d} \tilde T_1^{(h)}}{{\rm d}h} = -\frac{1}{h}\nu_h*f(X_1)
\eeqn
and
\beqn
\frac{{\rm d} \tilde z_1^{(h)}}{{\rm d}h} = -\frac{1}{h}\left[\eta_h*f(X_1)-\E\left\{\eta_h*f(X_1)\right\}\right].
\eeqn

Therefore,
\beqn
\frac{{\rm d} V_{nh}(X_1)}{{\rm d}h} = \frac{1}{h}\left[\tau_h*f(X_1)-\E\left\{\tau_h*f(X_1)\right\}\right],
\eeqn
where
\beqn
\tau(x) = 2K(x)+2xK'(x)-\frac{n-1}{n}\left\{K*K(x)+x(K*K)'(x)\right\}.
\eeqn

We have that
\beqn
\var\left\{\frac{{\rm d} V_{nh}(X_1)}{{\rm d}h}\right\} = \frac{1}{h}\left[\E\left\{\tau_h*f(X_1)^2\right\}-\E\left\{\tau_h*f(X_1)\right\}^2\right].
\eeqn

It is easy to show that
\beqn
\mu_0(\tau) &=& 0,\\
\mu_2(\tau) &=& -\frac{4}{n}\mu_2(K),\\
\mu_4(\tau) &=& -\frac{8}{n}\mu_4(K)+24\frac{n-1}{n}\mu_2(K)^2,\\
\mu_6(\tau) &=& -\frac{12}{n}\mu_6(K)+180\frac{n-1}{n}\mu_2(K)\mu_4(K).
\eeqn

Using standard calculations, one can see that
\beqn
\E\left\{\tau_h*f(X_1)\right\} &=& \int\int \tau(u)f(x)f(x-hu)\,{\rm d}x\,{\rm d}u\\
&=& \int\int \tau(u)f(x)\left\{f(x)-huf'(x)+\cdots-\frac{h^7u^7}{7!}f^{(7)}(x)+\frac{h^8u^8}{8!}f^{(8)}(\tilde x)\right\}\,{\rm d}x\,{\rm d}u\\
&=& -\frac{h^2}{2}\mu_2(\tau)R(f')+\frac{h^4}{24}\mu_4(\tau)R(f'')-\frac{h^6}{6!}\mu_6(\tau)R(f''')+O\left(h^8\right).
\eeqn

Therefore,
\beqn
\E\left\{\tau_h*f(X_1)\right\}^2 &=& \frac{h^4}{4}\mu_2(\tau)^2R(f')^2-\frac{h^6}{24}\mu_2(\tau)\mu_4(\tau)R(f')R(f'')\\
&+&\frac{h^8}{24^2}\mu_4(\tau)^2R(f'')^2+\frac{h^8}{6!}\mu_2(\tau)\mu_6(\tau)R(f')R(f''')+O\left(h^{10}\right).
\eeqn

On the other hand,
\beqn
\E\left\{\tau_h*f(X_1)^2\right\} &=& \int\int\int \tau(u)f(x-hu)\tau(v)f(x-hv)f(x)\,{\rm d}x\,{\rm d}u\,{\rm d}v\\
&=& \int\int\int \tau(u)\tau(v)f(x)\left\{f(x)-huf'(x)+\cdots+O\left(h^{10}\right)\right\}\\
&&\left\{f(x)-hvf'(x)+\cdots+O\left(h^{10}\right)\right\}\,{\rm d}x\,{\rm d}u\,{\rm d}v\\
&=& \frac{h^4}{4}\mu_2(\tau)^2J_2+\frac{h^6}{24}\mu_2(\tau)\mu_4(\tau)J_3+\frac{h^8}{6!}\mu_2(\tau)\mu_6(\tau)J_4+\frac{h^8}{24^2}\mu_4(\tau)^2J_1+O\left(h^{10}\right),
\eeqn
where
\beqn
J_2 &=& \int f(x)f''(x)^2\,{\rm d}x,\\
J_3 &=& \int f(x)f''(x)f^{(4)}(x)\,{\rm d}x,\\
J_4 &=& \int f(x)f''(x)f^{(6)}(x)\,{\rm d}x.
\eeqn

So, we have that
\beqn
\var\left\{\frac{{\rm d} V_{nh}(X_1)}{{\rm d}h}\right\} &=& \frac{h^2}{4}\mu_2(\tau)^2\left\{J_2-R(f')^2\right\}+\frac{h^4}{24}\mu_2(\tau)\mu_4(\tau)\left\{J_3+R(f')R(f'')\right\}\\
&+&\frac{h^6}{6!}\mu_2(\tau)\mu_6(\tau)\left\{J_4-R(f')R(f''')\right\}\\
&+&\frac{h^6}{24^2}\mu_4(\tau)^2\left\{J_1-R(f'')^2\right\}+O\left(h^8\right).
\eeqn

Finally, since
\beqn
\var\left\{S_1'(h)\right\} = \frac{4}{n}\var\left\{\frac{{\rm d} V_{nh}(X_1)}{{\rm d}h}\right\},
\eeqn
it follows that
\beqn
\var\left\{S_1'(h)\right\} = 4\mu_2(K)^4\left\{J_1-R(f'')^2\right\}\frac{h^6}{n}+O\left(\frac{h^8}{n}\right).
\eeqn

This, in conjunction with \eqref{eq:varSprimefinal}, proves that $\var\left\{S_1'(h)\right\}$ is negligible with respect to $\var\left\{S_2'(h)\right\}$ and, in particular, that $\var\left\{S_1'(h)\right\}$ cannot be asymptotic to $A_2\,h^2/n$ as claimed in \cite{HR}.

 

\end{document}